\documentclass{sig-alternate}
\newcommand{\REMARK}[2]{}
\newcommand{\daniel}[1]{\REMARK{Daniel}{#1}}
\newcommand{\yael}[1]{\REMARK{Yael}{#1}}
\newcommand{\val}[1]{\REMARK{Val}{#1}}
\newcommand{\finish}[1]{\REMARK{finish}{#1}}

\newcommand{\eat}[1]{}

\newcommand{\bbB}{\ensuremath{\mathbb{B}}}
\newcommand{\bbS}{\ensuremath{\mathbb{S}}}
\newcommand{\bbD}{\ensuremath{\mathbb{D}}}
\newcommand{\bbN}{\ensuremath{\mathbb{N}}}
\newcommand{\bbZ}{\ensuremath{\mathbb{Z}}}
\newcommand{\bbR}{\ensuremath{\mathbb{R}}}

\newcommand{\NX}{\ensuremath{\bbN[X]}}
\newcommand{\ZX}{\ensuremath{\bbZ[X]}}

\newcommand{\Nat}{\mbox{$\mathbb{N}$}}

\newcommand{\ssum}{\mbox{$\sum\,$}}
\newcommand{\notimes}{\mbox{$\otimes$}}
\newcommand{\sM}{\!_M}
\newcommand{\sK}{\!_K}
\newcommand{\sKprime}{\!_{K'}}
\newcommand{\sbbN}{\!_\bbN}
\newcommand{\sbbB}{\!_\bbB}
\newcommand{\sbbZ}{\!_\bbZ}

\newcommand{\sW}{\!_W}
\newcommand{\sKM}{\!_{K\!\otimes\!M}}
\newcommand{\sSS}{\!_{\bbS}}
\newcommand{\supp}{\mbox{supp}}
\newcommand{\Tuples}[2]{#2^{#1}}

\newcommand{\BoolExp}{\mbox{BoolExp}}

\newcommand{\Rel}[1]{#1\mbox{-Rel}}
\newcommand{\hRel}[1]{{#1}_{Rel}}

\newcommand{\sMaxAgg}{\!_{\mbox{\tiny MAX}}}

\newcommand{\SumAgg}{\mbox{SUM}}
\newcommand{\SumAggWithSpace}{\mbox{SUM} }
\newcommand{\ProdAgg}{\mbox{PROD}}
\newcommand{\MaxAgg}{\mbox{MAX}}

\newcommand{\MinAgg}{\mbox{MIN}}

\newcommand{\Agg}{\mbox{AGG}}
\newcommand{\SPJUlastA}{\mbox{SPJU-A} }
\newcommand{\Set}[1]{#1\mbox{-Set}}

\newcommand{\SPJUlastAGB}{\mbox{SPJU-AGB} }

\newcommand{\SetAgg}{\mbox{SetAgg}}

\newtheorem{theorem}{Theorem}[section]
\newtheorem{lemma}[theorem]{Lemma}
\newtheorem{proposition}[theorem]{Proposition}
\newtheorem{definition}[theorem]{Definition}
\newtheorem{corollary}[theorem]{Corollary}
\newtheorem{example}[theorem]{Example}

\newcommand{\bdefine}{\begin{Definition}}

\newcommand{\edefine}{\unitend\end{Definition}}
\newtheorem{Definition}{Definition}[section]
\newcommand{\unitend}{\hfill$\boxtimes$}

\newcommand{\NXD}{\ensuremath{\bbN[X,\delta]}}
\newfont{\affil}{phvr8t at 8pt}

\title{Provenance for Aggregate Queries}
\author{ Yael Amsterdamer \\[1mm] {\affil University of Pennsylvania} \\ {\affil and Tel Aviv University}  \and Daniel Deutch \\[1mm] {\affil University of Pennsylvania} \\ {\affil and Ben Gurion University} \and Val Tannen \\[1mm] {\affil University of Pennsylvania} }
\begin{document}
\maketitle

\begin{abstract}
We study in this paper provenance information for queries with {\em aggregation}. Provenance information was studied in the context of various query languages that do not allow for aggregation, and recent work has suggested to capture provenance by annotating the different database tuples with elements of a {\em commutative semiring} and propagating the annotations through query evaluation. We show that aggregate queries pose novel challenges rendering this approach inapplicable. Consequently, we propose a new approach, where we annotate with provenance information not just tuples but also the {\em individual values} within tuples, using provenance to describe the values computation.  We realize this approach in a concrete construction, first for ``simple" queries where the aggregation operator is the last one applied, and then for arbitrary (positive) relational algebra queries with aggregation; the latter queries are shown to be more challenging in this context. Finally, we use aggregation to encode queries with {\em difference}, and study the semantics obtained for such queries on provenance annotated databases.
\end{abstract}

\section{Introduction}

%%%%%
%%Provenance information is essential, e.g. for deletion propagation.
%%%%%

The annotation of the results of database transformations with
provenance information has quite a few applications \cite{Userssemiring1,Trio,whynot,CheneyTan,Userssemiring3,CheneyProvenance,Userssemiring4,Userssemiring5,Userssemiring6,Sigmod10Boon,Soul,FA,Archer}.
Recent work~\cite{GKT-PODS07,FGT-PODS08,G-ICDT09} has proposed
a framework of {\em semiring annotations} that allows us to state
formally what is expected of such provenance information. These papers
have developed the framework for the positive fragment of the relational
algebra (as well as for Datalog, the positive Nested Relational Calculus,
and some query languages on trees/XML). The main goal of this paper is to
extend the framework to {\em aggregate operations}.

In the perspective promoted by these papers, {\em provenance} is a
general form of annotation information that can be specialized for
different purposes, such as multiplicity, trust, cost, security, or
identification of ``possible worlds'' which in turn applies to
incomplete databases, deletion propagation, and probabilistic
databases. In fact, the introduction of the framework in~\cite{GKT-PODS07}
was motivated by the need to track trust and deletion propagation in the
Orchestra system~\cite{GKIT-VLDB07}.
What makes such a diversity of applications possible
is that each is captured by a different semiring, while provenance
is represented by elements of a semiring of {\em polynomials}.
One then relies on the property that any semiring-annotation
semantics {\em factors} through the provenance polynomials semantics.
This means that storing provenance polynomials
allows for many other practical applications. For example, to capture access control, where the access to different tuples require different
security credentials, we can simply evaluate the polynomials in the {\em security semiring}, and propagate the security annotations through query evaluation (see Section~\ref{semiringsubsec}), assigning security levels to query results.

\eat{
In designing the extension of the framework we must ask what
properties are expected of the formalism. One such criterion, and in
our view a basic one, has to do with the {\em propagation of
  deletion} which becomes particularly simple in the semiring framework
because it conflates the absence of a tuple
with the {\em zero} annotation.
}
\begin{figure}
$$
\begin{array}{|ccc||l|}
\multicolumn{4}{c}{R} \\
%\multicolumn{4}{c}{} \\
\hline
\mathit{EmpId} & \mathit{Dept} & \mathit{Sal} & \\
\hline
\hline
1 & \mathsf{d_1} & 20 & p_1 \\
2 & \mathsf{d_1} & 10 & p_2 \\
3 & \mathsf{d_1} & 15 & p_3 \\
4 & \mathsf{d_2} & 10 & r_1 \\
5 & \mathsf{d_2} & 15 & r_2 \\
\hline
\multicolumn{4}{c}{} \\
\multicolumn{4}{c}{(a)}
\end{array}
~~~~
\begin{array}{|c||l|}
%\multicolumn{4}{c}{R} \\
%\multicolumn{4}{c}{} \\
\hline
\mathit{Dept} & \\
\hline
\hline
\mathsf{d_1} & p_1+p_2+p_3 \\
\mathsf{d_2} & r_1+r_2 \\
\hline
              \multicolumn{2}{c}{} \\
\multicolumn{2}{c}{(b)}
\end{array}
$$
\vspace{-6mm}
\caption{Projection on annotated relations}
\label{AnnotProj}
\vspace{-6mm}
\end{figure}
Let us briefly illustrate {\em deletion propagation} as an application of
provenance. Consider a simple example of an employee/department/salary
relation $R$ shown in Figure~\ref{AnnotProj}(a).

\eat{(This is a rank/salary structure at the Free University of Fredonia
and the salaries are expressed in the local currency, kilo-shyssels.)}
The variables $p_1,p_2,p_3,r_1,r_2$ can be thought of as tuple identifiers
and in the framework of {\em provenance polynomials}~\cite{GKT-PODS07}
they are the ``provenance tokens'' or ``indeterminates'' out of which
provenance is built.
We denote by $\NX$ the set of provenance polynomials
(here $X=\{p_1,p_2,p_3,r_1,r_2\}$). $R$ can be seen as an $\NX$-annotated
relation; as defined in \cite{GKT-PODS07} the evaluation of query, for example $\Pi_{\mathit{Dept}}R$, produces another
$\NX$-annotated relation, in this example the one shown in
Figure~\ref{AnnotProj}(b). Intuitively, in this simple example, the summation in the annotation of every result tuple
is over the identifiers of its alternative origins \footnote{we explain how the annotations of query results are computed in Section \ref{semiringsubsec}}.

Now, the result of propagating the deletions of tuples with {\it EmpId}
3 and 5 in $R$ is obtained by simply setting $p_3=r_2=0$ in the
answer. We get the same two tuples in the query answer but their provenances change to
$p_1+p_2$ and $r_1$, respectively.  If the tuple with {\it EmpId} 4 is also
deleted from $R$ then we also set $r_1=0$, and the second tuple in the
answer is deleted because its provenance has now become 0. This
algebraic treatment of deletions is related to the counting
algorithm for view maintenance~\cite{GMS-SIGMOD93}, but is more
general as it incrementally maintains not just the data but also
the provenance.

An intuitive way of understanding what happens is that
provenance-aware evaluation of queries conveniently ``commutes'' with
deletions. In fact, in~\cite{GKT-PODS07,FGT-PODS08} this intuition is captured
formally by theorems that state that query evaluation commutes with semiring
homomorphisms. The factorization through provenance relies on this
and on the fact that the polynomial provenance semiring is ``freely
generated''. All applications of provenance polynomials we have listed,
for trust, security, etc., are based on these theorems.

Thus, commutation with homomorphisms is an essential criterion
for our proposed framework extension to aggregate operations.
However, in Section \ref{FirstAttempt} we prove that
the framework of semiring-annotated relations introduced in~\cite{GKT-PODS07}
cannot be extended to handle aggregation while both satisfying commutation
with homomorphisms and working as usual on set or bag relations.

\eat{Intuitively, this is because this framework can capture only the results of monotone queries.}

\eat{
We use the
particular case of deletion propagation to illuminate the basic
difficulties raised by this extension.
}

If the semiring operations are not enough then perhaps we can
add others? This is a natural idea so we illustrate it on the same $R$
in Figure~\ref{AnnotProj}(a) and we use again the necessity to
support deletion propagation to guide the approach.
Consider the query that groups by {\it Dept} and sums {\it Sal}.
The result of the summation depends on which tuples participate in it.
To provide enough information to obtain all the
possible summation results for all possible sets of deletions, we could
use the representation in Figure~\ref{ExpSize}(a)
\begin{figure}
\vspace{-5mm}
$$
\begin{array}{|cc||l|}
%\multicolumn{4}{c}{R} \\
%\multicolumn{4}{c}{} \\
\hline
\mathit{Dept} & \mathit{SalMass} & \\
\hline
\hline
\mathsf{d_1} & 45 & p_1p_2p_3 \\
\mathsf{d_1} & 30 & p_1p_2\widehat{p_3} \\
\mathsf{d_1} & 35 & p_1\widehat{p_2}p_3 \\
\mathsf{d_1} & 25 & \widehat{p_1}p_2p_3\\
\mathsf{d_1} & 20 & p_1\widehat{p_2}\widehat{p_3} \\
\mathsf{d_1} & 10 & \widehat{p_1}p_2\widehat{p_3} \\
\mathsf{d_1} & 15 & \widehat{p_1}\widehat{p_2}p_3 \\
    \cdots & \cdots & \cdots \\
\hline
\multicolumn{3}{c}{} \\
\multicolumn{3}{c}{(a)}
\end{array}
~~~~
\begin{array}{|cc||l|}
%\multicolumn{4}{c}{R} \\
%\multicolumn{4}{c}{} \\
\hline
\mathit{Dept} & \mathit{SalMass} & \\
\hline
\hline
\mathsf{d_1} & 30 & p_1p_2 \\
\mathsf{d_1} & 20 & p_1\widehat{p_2} \\
\mathsf{d_1} & 10 & \widehat{p_1}p_2 \\
    \cdots & \cdots & \cdots \\
\hline
\multicolumn{3}{c}{} \\
\multicolumn{3}{c}{(b)}
\end{array}
$$
\vspace{-5mm}
\caption{A naive approach to aggregation}
\vspace{-5mm}
\label{ExpSize}
\end{figure}
where we add to the semiring operations an unary
operation~~$\widehat{}$~~
with the property that $\widehat{p} = 1$ whenever $p=0$.
This will indeed satisfy the deletion criterion. For example
when the tuple with {\it Id} 3 is deleted we get the relation in
Figure~\ref{ExpSize}(b). In fact, there exist semirings with the
additional structure needed to define~~$\widehat{}$~. For example
in the semiring of polynomials with {\em integer} coefficients,
$\ZX$, we can take $\widehat{p} = 1-p$
while in the semiring
of boolean expressions with variables from $X$, $\BoolExp(X)$,
we can take $\widehat{p} = \neg p$
\footnote{Both of these also give natural semantics to relational
{\em difference}. $\ZX$ is used in~\cite{G-Thesis} following the
use of $\bbZ$ in~\cite{GIT-ICDT09}, while $\BoolExp(X)$ is used
in the seminal paper~\cite{IL84}.}.
The latter is essentially the approach taken
in~\cite{DBLP:journals/jiis/LechtenborgerSV02}.
However, whether we use $\ZX$ or $\BoolExp(X)$, we still
have, in the worst case, {\em exponentially many different results}
to account for, at least in the case of summation
(a lower bound recognized also
in~\cite{DBLP:journals/jiis/LechtenborgerSV02}).
It follows that summation in particular (and therefore
any uniform approach to aggregation) cannot be represented with
a feasible amount of annotation as long as the annotation
stays at the tuple level.

Instead, we will present a provenance representation for
aggregation results that leads only to a {\em poly-size
  increase}, one that we believe is tractably implementable using methods
similar to the ones used in Orchestra~\cite{GKIT-VLDB07}.
We achieve this via a more radical approach:
we annotate with provenance information not
just the tuples of the answer but also {\em the manner in which the
  values in those tuples are computed}.

%\eat{
%The formalism is in fact suggested by the basic semiring
%annotation propagation~\cite{GKT-PODS07}. After all, even for a SPJU
%query there are exponentially many different sets of deletions that need to be
%accounted for, and still annotating each answer tuple with a provenance
%polynomial from $\NX$ provides for SPJU
%queries polynomial-size information that accounts for all possible sets of
%deletions~\cite{GKT-PODS07}. This works because
%SPJU queries are {\em monotone} and deletions in
%the input change the output only by (possibly) deleting tuples.
%Meanwhile, aggregation, seen as a transformation between relations
%is not monotone because deletions in the input
%{\em update} the tuples of output (hence the impractical
%negative information approach discussed above). But seen
%as a transformation between relations and a formal
%representation of its {\em computation} an aggregation is monotone.
%}
%
%\daniel{I've eaten a paragraph above, dont think it helps here}

We can gain intuition towards our representation from the
particular case of bags,
which are in fact $\bbN$-relations, i.e., relations whose tuples are
annotated with elements of the semiring $(\bbN,+,\cdot,0,1)$. Assume
that $R$ in Figure~\ref{AnnotProj}(a) is such a relation, i.e.,
$p_1,\ldots,r_2\in\bbN$ are tuple multiplicities. Then, after
sum-aggregation the value of the attribute {\it SalMass} in the
tuple with {\it Dept} {\sf d}$_1$ is computed by
$p_1\times 20 + p_2 \times 10 +
p_3\times 15$. Now, if the multiplicities are, for example
$p_1=2,p_2=3,p_3=1$ then the aggregate value is 85. But what if $R$ is
a relation annotated with provenance polynomials rather than multiplicities?
Then, the aggregate value does not correspond to any number.

We will make $p_1\times 20$ into an {\em abstract} construction, $p_1\otimes 20$
and the aggregate value will be the formal expression
$p_1\otimes 20 + p_2\otimes 10 + p_3\otimes 15$.

Intuitively, we are embedding the domain of sum-aggregates, i.e., the reals
$\bbR$, into a larger domain of formal expressions that capture how the
sum-aggregates are computed from values annotated with provenance. We do
the same for other kinds of aggregation, for instance min-aggregation
gives $p_1\otimes 20 \min p_2\otimes 10 \min p_3\otimes 15$.
We call these {\em annotated aggregate} expressions.

In this paper we consider only aggregations
defined by commutative-associative operations
\footnote{As shown in~\cite{DBLP:conf/ctcs/LellahiT97}, for list
collections it also makes sense to consider non-commutative aggregations.}.
Specifically, our framework can accommodate
aggregation based on any {\em commutative monoid}. For example the commutative
monoid for summation is $(\bbR,+,0)$ while the one for min is
$(\bbR^\infty,\min,\infty)$\footnote{$K$-annotated relations with
{\em union} (see Section \ref{semiringsubsec}) also form such a structure.}.

%Just as with the relational algebra on semiring-annotated
%relations~\cite{GKT-PODS07}, the semantics of annotated aggregation
%must satisfy certain {\em equivalence laws}.
%For instance, if $A,B$ are finite sets of real numbers then
%$\min(A\cup B) = \min(\min(A),\min(B))$ (but
%$\sumagg(A) + \sumagg(B) = \sumagg(A\uplus B)$
%where $\uplus$ is {\em bag} union.
%%more on this in \finish{forward pointer}.
%\eat{
%To capture the relationship between min-aggregation
%and (set) union we need $p\otimes 30 \min r\otimes 30 = (p+r)\otimes 30$. \daniel{Did not completely understand this}.
%Similarly for sum-aggregation and bag-union.
%}
%We identify precisely these laws by showing that the right algebraic
%structure for capturing aggregates over $K$-relations, where $K$ is a
%commutative semiring, is that of $K${\em -semimodules}
%\val{We are kind of promising that in sec 2.2 we will justify
%why semimodules. We didn't yet...}
To combine an aggregation monoid $M$ with an annotation commutative semiring
$K$, in a way capturing
aggregates over $K$-relations, we propose the use of the algebraic structure of $K${\em -semimodules} (see Section~\ref{Semimodules}).
Semimodules are a way of generalizing (a lot!) the operations considered
in linear algebra. Its ``vectors'' form only a commutative monoid (rather
than an abelian group) and its ``scalars'' are the elements of $K$ which
is only a commutative semiring (rather than a field).

In general, a commutative monoid $M$ does not have an obvious
structure of $K$-semimodule. To make it such we may need to add new
elements corresponding to the scalar multiplication of elements of $M$
with elements from $K$, thus ending up with the formal expressions
that represent aggregate computations, motivated above, as elements of
a {\em tensor product} construction $K\otimes M$.  We show that the
use of tensor product expressions as a formal representation of
aggregation result is effective in managing the
provenance of ``simple aggregate" queries, namely queries where the
aggregation operators are the last ones applied.

%\val{I think we should say something about $M$ and $K$ compatibility here}
%\daniel{

We show that certain semirings are ``compatible" with certain monoids, in the
sense that the results of computation done in $K \otimes M$ may be mapped to $M$, faithfully representing the aggregation
results. Interestingly, compatibility is aligned with common wisdom: it is known that some (idempotent) aggregation functions such as MIN and MAX
work well for set relations, while SUM and PROD require the treatment of relations as bags. We show that
non-idempotent monoids are compatible only with ``bag" semirings, i.e. semirings from which there exists an homomorphism to $\bbN$.

%}

In general, aggregation results may be used by the query as the input to further
operators, such as value-based joins or selections. Here the formal representation
of values leads to a difficulty: the truth values of comparison
operators on such formal expressions is undetermined! Consequently, we
extend our framework and construct semirings in which formal
comparison expressions between elements of the corresponding
semimodule are elements of the semiring.  This means that an
expression like $[p_1\otimes 20 = p_2\otimes 10 + p_3\otimes 15]$ may
be used as part of the provenance of a tuples in the join result.
%\val{Let's try to avoid this in the intro, just the idea [P(e1,e2)]}
This
expression is simply treated as a new provenance token
(with constraints), until $p_1,p_2,p_3$ are
assigned e.g. values from $\bbB$ or $\bbN$, in which case we can interpret both sides of the equality as elements of the monoid
and determine the truth value of the equality (see Section \ref{GeneralAggSec}). We show in Section~\ref{GeneralAggSec} that this
construction allows us to manage provenance information for arbitrary
queries with aggregation, while {\em keeping the representation size
  polynomial in the size of the input database}.  Our construction is
robust: if further queries are applied, the token $[p_1\otimes
  20 = p_2\otimes 10 + p_3\otimes 15]$ can be used as part of a more complex expression, just as any
other provenance token.

The main result of this paper is providing, for the first time, a semantics
for aggregation (including group by) on semiring-annotated
relations that:
\vspace{-1mm}
\begin{itemize}
\item Coincides with the usual semantics on set/bag relations for min/max/sum/prod.
\vspace{-1mm}
\item Commutes with homomorphisms of semirings (hence all the ensuing applications)
\vspace{-1mm}
\item Is representable with only poly-size overhead.
\end{itemize}
\vspace{-1mm}
\eat{
This fills one of the \val{gaps??} left open by~\cite{GKT-PODS07} and
the follow-up papers.
Another \val{such gap} is the treatment of query operations that
involve forms of negation.
}
A second result of this paper is a new semantics for {\em difference}
on relations annotated with elements from {\em any}
commutative semiring. This is done via an encoding of relational difference
using nested aggregation. The fact that such an encoding can be done is known
(see e.g.~\cite{K-PODS10,BNTW95}), but
combined with our provenance framework, the encoding gives a semantics for
``provenance-aware" difference. Our new semantics for $R-S$ is a hybrid of
bag-style and set-style semantics, in the sense that tuples
of $R$ that appear in $S$ do not appear in $R-S$ (i.e. a boolean
negative condition is used), while those that do not appear in $S$
appear in $R-S$ with the same annotation (multiplicity, if $K= \bbN$
is used) that they had in $R$. This makes the semantics different from
the bag-monus semantics and its generalization to ``monus-semirings"
in ~\cite{Userssemiring1} as well as from the
``negative multiplicities'' semantics in~\cite{GIT-ICDT09}
(more discussion in Section~\ref{related}). We examine the equational
laws entailed by this new
semantics, in contrast to those of previously proposed
semantics for difference. In our opinion, this semantics is probably not
the last word on difference of annotated relations, but we hope that
it will help inform and calibrate future work on the topic.

%\val{Maybe in teh section on difference we should say something about monotonicity and lack therefof}

\paragraph*{Paper Organization} The rest of the paper is organized as follows. Section \ref{prelim} describes and exemplifies the main mathematical ingredients used throughout the paper. Section \ref{SimpleAggSection} describes our proposed framework for ``simple" aggregation queries, and this framework is extended in Section \ref{GeneralAggSec} to nested aggregation queries. We consider difference queries in Section \ref{DifferenceSection}. Related Work is discussed in Section \ref{related}, and we conclude in Section \ref{ConcSection}.

%
%\finish{MORE STUFF BASED ON THE FOLLOWING ITEMS}
%
%\begin{itemize}
%\item Re-evaluating the query for every deletion may not be an option, for instance in cases where the query itself is no longer available, or in context of data warehousing where the data is collected from various sources.
%
%\item But even this is not good enough when the query may compare aggregation results (i.e. join on them, select based on them etc.)
%\item The difficulty is fundamental given the annotations used so far (N[X], Z[X]), and is independent on our construction of how provenance information propagates through query evaluation.
%\item We show how by extending the domain of annotations to include equalities constructs. We prove that this way we are able to handle arbitrary aggregation queries, while keeping the provenance size polynomial
%\item We show that difference may be encoded as a query with aggregation, and consequently that provenance with difference may be handled similarly. We discuss the semantics obtained for difference using the encoding, and compare it to other available semantics.
%\item Finally, we consider other applications of the provenance additional to deletion propagation, such as general update propagation, certainty and possibility.
%\end{itemize}

%regation queries, while keeping the provenance size polynomial
%\item We show that difference may be encoded as a query with aggregation, and consequently that provenance with difference may be handled similarly. We discuss the semantics obtained f 
\vspace{-1mm}
\section{Preliminaries}
\label{prelim}

%\val{Pretty lame section title. see if we can do better...}

We provide in this section the algebraic foundations that will be used throughout the paper. We start by recalling the notion of semiring and its use in ~\cite{GKT-PODS07} to capture provenance for the SPJU algebra queries. We then consider aggregates, and show the new algebraic construction that is required to accurately support it.
\subsection{Semirings and SPJU}
\label{semiringsubsec}
We briefly review the basic framework introduced in~\cite{GKT-PODS07}.
A {\em commutative monoid} is an algebraic structures
$(M,+_{\sM},0_{\sM})$ where $+_{\sM}$ is an associative
and commutative binary operation and $0_{\sM}$ is an identity for $+_{\sM}$. A monoid homomorphism is a mapping
$h:M \rightarrow M'$ where $M,M'$ are monoids, and $h(0_{\sM})=0_{\!_{M'}}$,$h(a+b)=h(a)+h(b)$.
We will consider database operations on relations whose tuples
are annotated with elements from {\em commutative semirings}. These
are structures $(K,+_{\sK} ,\cdot_{\sK},0_{\sK},1_{\sK})$ where
$(K,+_{\sK} ,0_{\sK})$ and $(K,\cdot_{\sK},1_{\sK})$ are commutative monoids,
$\cdot_{\sK}$ is distributive over $+_{\sK} $, and
$a\cdot_{\sK}0_{\sK} = 0\cdot_{\sK} a = 0_{\sK}$.
A semiring homomorphism is a mapping
$h:K \rightarrow K'$ where $K,K'$ are semirings, and $h(0_{\sK})=0_{\!_{K'}},h(1_{\sK})=1_{\!_{K'}}$,$h(a+b)=h(a)+h(b)$, $h(a \cdot b)=h(a) \cdot h(b)$.
Examples of commutative semirings are any commutative ring (of course) but also
any distributive lattice, hence any boolean algebra.
%\cite{AlgebraBook}
Examples of particular interest to us include
the boolean semiring $(\bbB,\vee,\wedge,\bot,\top)$
(for usual set semantics),
the natural numbers semiring $(\bbN,+,\cdot,0,1)$
(its elements are multiplicities, i.e., annotations that give bag semantics),
and the {\em security}
semiring $(\bbS,\min,\max,0_{\sSS},1_{\sSS})$ where $\bbS$ is the
ordered set, $1_{\sSS}<\mathsf{C}<\mathsf{S}<\mathsf{T}<0_{\sSS}$
whose elements have the following meaning
when used as annotations: $1_{\sSS}$ : public (``always available"),
$\mathsf{C}$ : confidential, $\mathsf{S}$ : secret,
$\mathsf{T}$ : top secret, and $0_{\sSS}$ means ``never available".

%A semiring $K$ is positive if
%for every $a,b \in K$, $a+_{\sK}b =0_{\sK}$ if and only if $a=0_{\sK}$ and $b=0_{\sK}$. The above semirings are all examples of positive semirings, while
%$(\bbZ,+,\cdot,0,1)$ is a non-positive semiring.

Certain semirings play an essential role in capturing provenance
information.  Given a set $X$ of {\em provenance tokens} which
correspond to ``atomic'' provenance information, e.g., tuple identifiers,
the semiring of {\em polynomials} $(\NX,+,\cdot,0,1)$ was shown
in~\cite{GKT-PODS07} to adequately, and most generally,
capture provenance for positive relational queries.
The {\em provenance interpretation} of the semiring structure
is the following. The $+$ operation on annotations corresponds
to {\em alternative use}
of data, the $\cdot$ operation to {\em joint use} of data, $1$ annotates
data that is always and unrestrictedly available, and $0$ annotates absent data. The definition
of the $K$-relational algebra (see bellow for union, projection and join)
fits indeed this interpretation.
Algebraically, $\NX$ is the commutative semiring freely generated by $X$,
i.e., for any other commutative semiring $K$, any valuation of the
provenance tokens $X\rightarrow K$ extends uniquely to a semiring
homomorphism $\NX\rightarrow K$ (an evaluation in $K$ of the polynomials).
We say that any semiring annotation semantics {\em factors} through
the provenance polynomials semantics, which means that for practical
purposes storing provenance information suffices for many other applications
too. Other semirings can also be used to capture certain forms of provenance,
albeit less generally than $\NX$~\cite{GKT-PODS07,G-ICDT09}. For example,
boolean expressions capture enough provenance to serve in the
intensional semantics of queries on
incomplete~\cite{IL84} and probabilistic data~\cite{FR97,Z97}.

To define annotated relations we use here the named perspective of the
relational model~\cite{AHV}. Fix a countably infinite
\emph{domain} $\bbD$ of values (constants).
For any finite set $U$ of attributes a tuple is a function
$t : U \rightarrow \bbD$ and we denote the set of all such possible tuples
by $\Tuples{U}{\bbD}$. Given a commutative semiring $K$, a
{\em $K$-relation} (with schema $U$) is a function
$R: \Tuples{U}{\bbD}\rightarrow K$ whose
\emph{support}, $\supp(R) = \{t \mid R(t) \neq 0_{\sK}\}$ is
\emph{finite}. For a fixed set of attributes $U$ we denote by $\Rel{K}$
(when $U$ is clear from the context) the set of $K$-relations
with schema $U$. We also define a $K${\em -set} to be
a function $S: \bbD\rightarrow K$ again of finite support. We then define:
\vspace{-1mm}
\begin{description}
\item[Union] If $R_{i}:\Tuples{U}{\bbD} \rightarrow K,~i=1,2$
then $R_1\cup_{\sK} R_2:\Tuples{U}{\bbD}\rightarrow K$
is defined by $(R_1\cup_{\sK}  R_2)(t) = R_1(t)+_{\sK}R_2(t)$.
The definition of union of $K$-sets follows similarly.
\end{description}
\vspace{-1mm}
We also define the {\em empty} $K$-relation ($K$-set)
by $\emptyset_{\sK} (t)=0_{\sK} $.
It is easy to see that
$(\Rel{K},\cup_{\sK} ,\emptyset_{\sK})$ is a commutative monoid
\footnote{In fact, it also has a semiring structure.}.% but the important structure on $\Rel{K}$ is that of $K$-semimodule, see Section \ref{Semimodules}}.
Similarly, we get the commutative monoid of $K$-sets
$(\Set{K},\cup_{\sK} ,\emptyset_{\sK})$.

Given a named relational schema, $K$-databases
are defined from $K$-relations just as relational databases are defined
from usual relations, and in fact the usual (set semantics) databases
correspond to the particular case $K=\bbB$. The (positive) {\em $K$-relational
algebra} defined in~\cite{GKT-PODS07} corresponds to a semantics on
$K$-databases for the usual operations of the relational algebra. We
have already defined the semantics of union above
and we give here just two other cases leaving the rest for Appendix
\ref{ProvenanceSemirings} (for a tuple $t$ and an attributes set $U'$, $t|_{U'}$ is the restriction of $t$ to $U'$):
\begin{description}
\item[Projection] If $R:\Tuples{U}{\bbD} \rightarrow K$
and $U' \subseteq U$ then
$\Pi_{U'}R:\Tuples{U'}{\bbD} \rightarrow K$ is defined
by  $(\Pi_{U'}R)(t) = \sum\sK  R(t')$ where the $+_{\sK}$ sum is over all
$t'\in\supp(R)$ such that $t'|_{U'}=t$.
\item[Natural Join] If $R_{i}:\Tuples{U_i}{\bbD} \rightarrow K,~i=1,2$
then $R_1\bowtie R2:\Tuples{U_1\cup U_2}{\bbD}\rightarrow K$ is defined by
$(R_1 \bowtie R_2)(t) = R_1(t_1) \cdot_{\sK}  R_2(t_2)$ where $t_i=t|_{U_i},~i=1,2$.
\end{description}

\eat{
\val{DO WE STILL NEED THIS:
A \emph{query} is then a mapping between $K$-databases, and
finally an \emph{algebra} for queries in a given class $C$ assigns
such a mapping to every syntactic representation of a query in
$C$.}\yael{can we embed into the definition of a $K$-database algebra the
  connection to the standard query algebra?}
}

\subsection{Semimodules and aggregates}
\label{Semimodules}
We will consider aggregates defined by commutative monoids.  Some
examples are $\SumAgg=(\bbR,+,0)$ for summation
\footnote{COUNT is particular case of summation and AVG is obtained
from summation and COUNT.},
$\MinAgg=(\bbR^{\pm\infty},\min,+\infty)$ for min,
$\MaxAgg=(\bbR^{\pm\infty},\max,-\infty)$ for max,
and $\ProdAgg=(\bbR,\times,1)$ for product.

In dealing with aggregates we have to extend the operation of a
commutative monoid to operations on relations annotated with elements
of semirings. This interaction will be captured by {\em semimodules}.
%\vspace{-2mm}
\begin{definition}
\label{SemiModuleDef}
Given a commutative semiring $K$,
a structure $(W,+_{\sW},0_{\sW},*_{\sW})$ is a $K$-semimodule if
$(W,+_{\sW},0_{\sW})$ is a commutative monoid and $*_{\sW}$ is a
binary operation $K\times W\rightarrow W$ such that
(for all $k,k_1,k_2 \in K$ and $ w,w_1,w_2 \in W$):
\begin{eqnarray}
    k*_{\sW}(w_1+_{\sW}w_2) & = & k*_{\sW}w_1 +_{\sW} k*_{\sW}w_2 \\
           k*_{\sW}0_{\sW} & = & 0_{\sW} \\
    (k_1+_{\sK}  k_2)*_{\sW}w & = & k_1*_{\sW}w +_{\sW} k_2*_{\sW}w \\
           0_{\sK} *_{\sW}w & = & 0_{\sW} \\
(k_1\cdot_{\sK} k_2)*_{\sW}w & = & k_1*_{\sW}(k_2*_{\sW}w) \\
           1_{\sK} *_{\sW}w & = & w
\end{eqnarray}

\end{definition}
%\vspace{-1mm}
In any (commutative) monoid $(M,+_{\sM},0_{\sM})$
define for any $n\in\bbN$ and $x\in M$
$$
nx = x+_{\sM} \cdots +_{\sM}x~~(n ~~ \mbox{times})
$$
in particular $0x=0_{\sM}$. Thus $M$ has a canonical
structure of $\bbN$-semimodule. Moreover, it is easy to check that
a commutative monoid $M$ is a $\bbB$-semimodule if and only if its
operation is {\em idempotent}: $x+_{\sM}x = x$. %\eat{
%\val{I think we should define here the extension of the relational
%algebra with aggregation, assuming that the domain of the aggregation
%has a scalar multiplication operation and then show the equivalence
%between a semimodule axioms and algebra laws. OR MAYBE NOT, see below}
%}
The $K$-relations themselves form a $K$-semimodule
$(\Rel{K},\cup_{\sK} ,\emptyset_{\sK} ,*_{\sK})$ where
$(k*\sK R)(t) = k\cdot_{\sK} R(t)$
\footnote{In fact, it is the $K$ semimodule freely generated
by $\Tuples{U}{\bbD}$.}.

We now show, for any $K$-semimodule $W$, how to define {\em
  $W$-aggregation} of a $K$-set of elements from $W$.  We assume that
$W\subseteq\bbD$ and that we have just one attribute, whose values are
all from $W$. Consider the $K$-set $S$ such that
$\supp(S)=\{w_1,\ldots,w_n\}$ and $S(w_i)=k_i\in K, i=1,\ldots, n$
(i.e., each $w_i$ is annotated with $k_i$). Then, the result of
$W$-aggregating $S$ is defined as
$$
\SetAgg_W(S) ~=~ k_1*_{\sW}w_1 +_{\sW}\cdots +_{\sW} k_n*_{\sW}w_n ~~\in W
$$
For the empty $K$-set we define $\SetAgg_W(\emptyset_{\sK}) = 0_{\sW}$.
Clearly, $\SetAgg_W$ is a {\em semimodule homomorphism}
\footnote{In fact, it is the free homomorphism
determined by the identity function on $W$.}.
Since all commutative monoids are $\bbN$-semimodules
this gives the usual sum, prod, min, and max aggregations on bags.
Since $\MinAgg$ and $\MaxAgg$ are $\bbB$-semimodules this gives
the usual min and max aggregation on sets
\footnote{The fact that the right algebraic structure to use
for aggregates is that of {\em semimodules} can be justified in the
same way in which using semirings was justified in~\cite{GKT-PODS07}:
by showing how the laws of semimodules follow from
desired equivalences between aggregation queries.}.

Note that $\SetAgg$ is an operation on sets, not an operation on relations.
In the sequel we show how to extend it to one.
%\val{and how the semimodule axioms are determined by
%expected relational equivalence laws?}

\subsection{A tensor product construction}
\label{Tensor}

More generally, we want to investigate $M$-aggregation on
$K$-relations where $M$ is a commutative monoid and $K$ is some
commutative semiring. Since $M$ may not have enough elements to
represent $K$-annotated aggregations we construct a $K$-semimodule in which $M$ can be embedded, by
transferring to semimodules the basic idea behind a standard algebraic
construction, as follows.

Let $K$ be any commutative semiring and $M$ be any commutative monoid.
We start with $K\times M$, denote its elements
$k\notimes m$ instead of $\langle k,m\rangle$ and call them ``simple tensors''.
Next we consider (finite) bags of such simple tensors, which, with bag
union and the empty bag, form a commutative monoid.
It will be convenient to
denote bag union by $+_{\sKM}$, the empty bag by $0_{\sKM}$ and to
abuse notation denoting singleton bags by the unique element they contain.
Then, every non-empty bag of simple tensors can be written
(repeating summands by multiplicity)
$k_1\notimes m_1+_{\sKM}\cdots +_{\sKM}k_n\notimes m_n$. Now we define
$$
k *_{\sKM} \ssum k_i\notimes m_i  = \ssum (k\cdot_{\sK} k_i)\notimes m_i
$$
Let $\sim$ be the smallest congruence w.r.t. $+_{\sKM}$
and $*_{\sKM}$ that satisfies (for all $k,k',m,m'$):
\begin{eqnarray*}
   (k+_{\sK}  k')\notimes m & \sim & k\notimes m  +_{\sKM} k'\notimes m \\
        0_{\sK} \notimes m & \sim & 0_{\sKM}  \\
k \notimes (m +_{\sM} m') & \sim & k\notimes m  +_{\sKM} k\notimes m' \\
       k \notimes 0_{\sM} & \sim & 0_{\sKM}
\end{eqnarray*}
We denote by $K\otimes M$ the set of {\em tensors}
i.e., equivalence classes of bags of simple tensors modulo
$\sim$. We show in Appendix \ref{AppendixTensor} that $K\otimes M$ forms
a $K$-semimodule. \val{When I agreed that you take this out of this
section I assumed you will put in the Appendix, including the
universality property. Please do so.}

\eat{
Define $\iota:M\rightarrow K\notimes M$ such that
$\iota(m)=1_{\sK} \notimes m$. Every tensor is a linear
combination of simple tensors from $\iota(M)$. More precisely,
\begin{proposition}
\label{universalityTensor}
$K\otimes M$ is a $K$-semimodule that satisfies the
  ``universality'' property that
for any $K$-semimodule $W$ and any homomorphism of monoids
$f:M\rightarrow W$ there exists a unique homomorphism of
$K$-semimodules $f^*:K\notimes M \rightarrow W$ such that $f^*\circ\iota = f$.
\end{proposition}
\finish{proof in Appendix}
We also say that $K\otimes M$ is the $K$-semimodule {\em freely generated}
by $M$ (hence the ``most economical'' appellation).

An immediate consequence of Proposition~\ref{universalityTensor},
is the existence (and uniqueness) of
the ``lifting'' of a homomorphism of semirings
$h: K\rightarrow K'$ to a homomorphism of monoids
$h^M: K\otimes M \rightarrow K'\otimes M$.
Indeed $K'\otimes M$ becomes a $K$-semimodule via $h$
so we can define $h^M$ as the the unique homomorphism
of $K$-semimodules that by the proposition
extends $\iota':M\rightarrow K'\notimes M$. Note that
$$
h^M(\ssum k_i\notimes m_i) = \ssum h(k_i)\notimes m_i
$$
}

\paragraph*{Lifting homomorphisms} Given a homomorphism of semirings
$h: K\rightarrow K'$, and some commutative monoid $M$,  we can ``lift" $h$ to a homomorphism of monoids in a natural way. The lifted homomorphism is denoted
$h^M: K\otimes M \rightarrow K'\otimes M$ and defined by:
$$
h^M(\ssum k_i\notimes m_i) = \ssum h(k_i)\notimes m_i
$$

\section{Simple aggregation queries}
\label{SimpleAggSection}

In this section we begin our study of the ``provenance-aware"
evaluation of aggregate queries, focusing on ``simple'' such
queries, that is, queries in which aggregations are done last; for example, un-nested SELECT FROM WHERE GROUP BY queries.
This avoids the need to compare values which are the result
of annotated aggregations and simplifies the treatment.
The restriction is relaxed in the more general framework
presented in Section \ref{GeneralAggSec}.

The section is organized as follows. We list the desired features of a
provenance-aware semantics for aggregation, and first try to design a semantics with these features,
{\em without} using the tensor product construction, i.e. by simply working with $K$-relations as done in \cite{GKT-PODS07}. We show that
this is impossible. Consequently, we turn to semantics that are based on combining aggregation with values via the tensor product construction. We
propose such semantics that do satisfy the desired features, first for relational algebra with an additional AGG operator on relations (that
allows aggregation of all values in a chosen attributes, but no grouping); and then for GROUP BY queries.

\eat{
We first recall the definition of K-relations from \cite{ValIcdt08},
and observe the difficulties in applying this definition, or its simple extensions, in presence of aggregate queries. We then proceed to formally define
relations in which provenance annotations are combined with values, using the K-semimodule construction (def. \ref{}), show an algebra for queries with aggregates evaluated on such relations, and analyze its properties.
\val{rewrite this paragraph after section gets final shape}
}

\eat{
\begin{example}
\label{KMRelation}

Consider the following $(M, \NX)$ relation $R$:
%$$
%\begin{array}{|ccc||l|}
%%\multicolumn{4}{c}{R} \\
%%\multicolumn{4}{c}{} \\
%\hline
%\mathit{Id} & \mathit{Rank} & \mathit{Salary} & \\
%\hline
%\hline
%1 & \mathsf{a} & 20 & p_1 \\
%2 & \mathsf{a} & 10 & p_2 \\
%3 & \mathsf{b} & 10 & p_3 \\
%\hline
%\end{array}
%$$

$$
\begin{array}{|ccc||l|}
\hline
\mathit{Dept.} & \mathit{Sal.} & \\
\hline
\hline
\mathsf{d_1}  & 20 & r_1 \\
\mathsf{d_1}  & 10 & r_2 \\
\mathsf{d_2}  & 10 & r_3 \\
\hline
\end{array}
$$

The annotations of tuples are polynomials in $\NX$ on a set $X$ of ``basic annotations" \daniel{This may need further explanations?}. The set of attributes is $U = \{\mathit{Dept.}, \mathit{Sal.}\}$.
The aggregatable attribute is \textit{Sal.}, and its values are taken from a Monoid $(\bbR, +_{\sM},0_{\sM})$. Note that $+_{\sM}, 0_{\sM}$
are symbols that will be given semantics based on the aggregation function used in a query.

\end{example}
}

%\subsection{Desiderata for an SPJU+agg algebra
%and what definitely doesn't work}
\subsection{Semantic desiderata and first attempts}
\label{FirstAttempt}

We next explain the desired features of a
provenance-aware semantics for aggregation. To illustrate the difficulties and the need for a more complex construction,
we will first attempt to define a semantics on $K$-relations, without using the tensor product construction of Section \ref{Tensor}.

We consider a commutative semiring $K$ (e.g., $\bbB, \bbN, \NX, \bbS$,
etc.)  for tuple annotations and a commutative monoid $M$ (e.g.,
$\SumAgg=(\bbR,+,0), \ProdAgg=(\bbR,\times,1)$, $\MaxAgg=$ \\$(\bbR^{-\infty},\max,-\infty)$, $\MinAgg=(\bbR^{\infty},\min,\infty)$ etc.)  for aggregation. We will
assume that the elements of $M$ already belong to the database domain,
$M\subseteq\bbD$.

We have recalled the semantics of SPJU queries
in Section \ref{semiringsubsec}. Now we wish to add an $M$-aggregation
operation $\Agg$ on relations. We then denote by $\SPJUlastA$ the restricted class of queries consisting of
any SPJU-expression followed possibly by just one application of
$\Agg$. This corresponds to SELECT AGG(*) FROM WHERE queries
(no grouping).

For the moment, we do not give a concrete semantics to $\Agg_M(R)$,
allowing any possible semantics where the result of $\Agg_M(R)$ is a $K$-relation. We note that $\Agg_M(R)$ should be defined iff
$R$ is a $K$-relation with one attribute whose values are in $M$.

What properties do we expect of a reasonable semantics for
$\SPJUlastA$ (including, of course, a semantics for $\Agg_M(R)$)? A basic sanity check is
\begin{description}
\item[Set/Bag Compatibility]
The semantics coincides with the usual one when $K=\bbB$ (sets)
and $M=\MaxAgg$ or $\MinAgg$, and when $K=\bbN$ (bags) and $M=\SumAgg$ or $\ProdAgg$.
\end{description}
Note that we associate min and max with sets and sum and product
with bags. Min and max work fine with bags too, but we get the same
result if we convert a bag to a set (eliminate duplicates) and then
apply them. Sum and product (in the context of other operations
such as projection) require us to use bags semantics in order to
work properly. This is well-known, but our general approach sheds
further light on the issue by discussing such ``compatibility'' for arbitrary semirings and monoids
in Section~\ref{Compat}.

As discussed in the introduction, a fundamental desideratum
with many applications is {\em commutation with homomorphisms}.
Note that a semiring homomorphism $h:K \rightarrow K'$ naturally
extends to a mapping $\hRel{h}:\Rel{K}\rightarrow\Rel{K'}$
via $\hRel{h}(R)= h \circ R$ (i.e. the homomorphism is applied on the
annotation of every tuple), which then further extends to
   $K$-databases. With this, the second desired property is
\begin{description}
\item[Commutation with Homomorphisms]
Given any two commutative semirings $K,K'$
and any homomorphism $h:K \rightarrow K'$, for any query
$Q$, its semantics on $K$-databases and on $K'$-databases
satisfy ~$\hRel{h}(Q(D))=Q(\hRel{h}(D))$
for any $K$-database $D$.
\end{description}
It turns out that this property determines quite precisely the way
in which tuple annotations are defined. We say that the semantics of an
operation $\Omega$ on $K$-databases is {\em algebraically uniform}
with respect to the class of commutative semirings if
the annotations of the output $\Omega(D)$ are defined by
the same (for all $K$) $\{+_{\sK},\cdot_{\sK},0_{\sK},1_{\sK}\}$-expressions, where the elements in the expressions are
the annotations of the input $D$. One can see that the definition of the
SPJU-algebra is indeed algebraically uniform and was shown
in~\cite{GKT-PODS07,FGT-PODS08} to commute with homomorphisms.
The connection between the two properties is general (proof deferred to the Appendix):
\begin{proposition}
\label{uniform}
A semantics commutes with homomorphisms iff it is algebraically uniform.
\end{proposition}

\eat{
We say that an algebra for which the above three fundamental properties hold is {\em proper}. Indeed, in  \cite{ValPods07} and follow-up works, the authors showed algebras that bare these properties, for positive relational algebra queries ($RA^{+}$) as well as for other expressive classes of queries (Datalog, the positive Nested Relational Calculus,
and some query languages on XML).
}

%These properties are indeed fundamental
%for algebras on $K$-relations: set (or bag)-compatibility is a ``sanity check", assuring that the semantics conform to the ``expected"
%outcome. Polynomial representation of query results is of course necessary for practical deployment of the algebra Last, the importance of commutation with homomorphism was exemplified in the Introduction \daniel{Would another example here help?}.
%

After stating two of the desired properties, namely set/bag
compatibility and commutation with homomorphisms we can already show
that it is not possible to give a satisfactory semantics to the
SPJU-A algebra within the framework used in \cite{GKT-PODS07} for the SPJU-algebra.

\begin{proposition}
\label{noMK}
There is no $K$-relation semantics for $\MaxAgg$-(or $\MinAgg$-)aggregation
that is both set-compatible and commutes with homomorphisms. Similarly,
there is no $K$-relation semantics for $\SumAgg$-aggregation
that is both bag-compatible and commutes with homomorphisms.
\end{proposition}
\begin{proof}
%We prove the proposition for \MaxAgg. The proofs for \MinAggWithSpace and \SumAggWithSpace
%are similar.
Assume by contradiction the existence of such semantics.
Consider the $\NX$-relation $R$ with one attribute and two tuples with values $10$ and $20$,
with the corresponding tuple annotations being $x,y \in X$.
Let $R'$ be $\Agg_{\sMaxAgg}(R)$ according to the assumed semantics; $R'$ is also an $\NX$-relation. Because a tuple $t$ with a value 10 is
a possible answer to the \MaxAgg-aggregation (when we set
$y=0$) it must occur in $\supp(R')$. Let $p\in\NX$ be the annotation
of the tuple $t$ (having value $10$) in $R'$. By algebraic uniformity
the only variables that can occur
in $p$ are $x$ and $y$, and we consequently denote it $p(x,y)$.
Consider two homomorphisms $h',h'':\NX\rightarrow \bbB$
defined by $h'(x)=h'(y)=\top$ and
by $h''(x)=\top, h''(y)=\bot$.
Applying $\Agg_{\sMaxAgg}$ to $\hRel{h'}(R)$ and $\hRel{h''}(R)$
should, by set-compatibility, work as usual. Hence,
by commutation with homomorphisms
$h'(p)=\bot$ and $h''(p)=\top$.
Functions on $\bbB$ defined by polynomials in $\NX$ are monotone in each variable. But
$\bot = h'(p(x,y)) = p(h'(x),h'(y)) = p(\top,\top)$
and
$\top = h''(p(x,y)) = p(h''(x),h''(y)) = p(\top,\bot)$, in contradiction to the monotonicity.

\end{proof}

Alternatively, one may consider going beyond semirings, to algebraic
structures with additional operations. We have briefly
explored the use of ``negative'' information in the introduction.
As we show there, one could use the ring structure on $\ZX$ (the
additional subtraction operation) or the boolean algebra structure
on $\BoolExp(X)$ (the additional complement operation) but the use
of negative operation does not avoid the need to enumerate in
separate tuples of the answer all the possible aggregation results
given by subsets of the input. In the case of summation, at least,
there are exponentially many such tuples. We reject such an approach
and we state as an additional desideratum:
\begin{description}
\item[Poly-Size Overhead]
For any query $Q$ and database $D$, the size of $Q(D)$, including annotations,
should be only polynomial in the size of $D$.
\end{description}

\eat{
\begin{theorem}
There is no $\ZX$ ($Bool$)-proper algebra on $(M,K)$-relations.
\end{theorem}

\begin{proof}
The intuition for the correctness of this proof can be found in the ``naive'' aggregation representation depicted in Figure \ref{ExpSize}. This example shows that the number of possible aggregation values may be of exponential size.
We prove it here formally: assume that we are aggregating over a group of tuples where the aggregated value of the first tuple is $1$, the aggregated value of the second tuple is $2^1$, and so on until $2^{n}$. Then for each different subset of tuples (except the empty set) we get a different aggregation result, i.e. there are $2^n-1$ possible aggregation results. In order to be set-compatible, any algebra on $(M,K)$-relations would have to generate a tuple for each possible aggregation result; thus its output would be exponential in the database size, and violate property (1) of proper algebras.
\end{proof}
}

We shall next show a semantics
to the \SPJUlastA-algebra that satisfies all three properties
we have listed.

%\subsection{Combining annotations with values and the \SPJUlastA-algebra}
\subsection{Annotations $\otimes$ values and \SPJUlastA}
\label{FirstAlgebra}
Let us fix a commutative monoid $M$ (for aggregation) and a commutative
semiring $K$ (for annotation). The {\em inputs} of our queries are as before: $K$-databases whose domain $\bbD$ includes the values $M$ over
which we aggregate. However, the outputs are more
complicated. The basic idea for the semantics of aggregation was
already shown in Section~\ref{Semimodules} where it is assumed that the
domain of aggregation has a $K$-semimodule structure. %However, not every $M$ has a semimodule structure with respect to every $K$.Fortunately, a
As we have shown in Section~\ref{Tensor}, we can give a
tensor product construction that embeds $M$ in the $K$-semimodule
$K\otimes M$ (note that this embedding is not always {\em faithful}, as discussed in Section \ref{Compat}).

For the output relations of our algebra queries, we thus need results of aggregation (i.e., the elements of $K\otimes M$)
to also be part of the domain out of which the tuples are
constructed. Thus for the output domain we will assume that
$K\otimes M\subseteq\bbD$, i.e. the result ``combines annotations with values''.
The elements of $M$ (e.g., real numbers for sum or max aggregation)
are still present, but only via the embedding  $\iota:M\rightarrow K\otimes M$ defined by $\iota(m)=1_{\sK} \notimes m$.

Having annotations from $K$ appear in the values will change the way in which we apply homomorphisms
to query results, so to emphasize the change we will call
$(M,K)${\em -relations} the $K$-annotated relations over such that the data domain $\bbD$
that includes $K\otimes M$. To summarize, the semantics of the
\SPJUlastA-algebra will map databases of $K$-relations
(with $M\subseteq \bbD$) to $(M,K)$-relations
(with $K\otimes M\subseteq \bbD$).

As we define the semantics of the \SPJUlastA-algebra, we first note
that for selection, projection, join and union the definition is the
same as for the SPJU-algebra on $K$-databases. The last step of
the query is aggregation, denoted $\Agg_M(R)$, and is well-defined iff $R$ is a $K$-relation with one attribute
whose values are in the $M$ subset of $\bbD$. To apply the
definition that uses the semimodule structure
(shown in Section~\ref{Semimodules}), we convert $R$ to an
$(M,K)$-relation $\iota(R)$ by replacing each $m\in M$ with
$\iota(m)=1_{\sK} \notimes m\in K\otimes M$. Then, if
$\supp(R)=\{m_1,\ldots,m_n\}$ and $R(m_i)=k_i\in K, i=1,\ldots, n$
(i.e., each $m_i$ is annotated with $k_i$) we define
$\Agg_M(R)$  as a one-attribute relation with one
tuple annotation is $1_{\sK}$ and whose content is $\SetAgg_{K\otimes M}(\iota(R))$, which is equal to
\begin{eqnarray*}
k_1*_{\sKM}\iota(m_1) +_{\sKM}\cdots +_{\sKM} k_n*_{\sKM}\iota(m_n) \\
= k_1\notimes m_1+_{\sKM}\cdots +_{\sKM}k_n\notimes m_n
\end{eqnarray*}
We define the annotation of the only tuple in the
output of $\Agg_M$ to be $1_{\sK}$, since this tuple is always available. However, the {\em content} of this tuple does depend on $R$.
For example, even when $R$ is empty the output is not empty:
by the semimodule laws, its content
is $0_{\sKM}=\iota(0_{\sM})$.

\paragraph*{Commutation with Homomorphisms} We have explained in Section \ref{Tensor} how to lift a homomorphism
$h:K \rightarrow K'$ to a homomorphism $h^{M} : K \otimes M \rightarrow K' \otimes M$. Via this we can lift $h$ to a homomorphism $\hRel{h}$
on $(M,K)${\em -relations}: let $R$ be such a relation and recall that some values in $R$ are elements of $K \otimes M$, and the annotations of these tuples are elements of $K$.
Then $\hRel{h}(R)$ denotes the relation obtained from $R$ by replacing every $k \in K$ with $h(k)$,
and additionally replacing every $k \otimes M \in K \otimes M$ with $h^{M}(k \otimes m)$. All other values in $R$ stay intact. Applying $\hRel{h}$ on
a $(M,K)$-database $D$ amounts to applying $\hRel{h}$ on each $(M,K)$-relation in $D$.

We can now state the main result for our \SPJUlastA-algebra:

\begin{theorem}
\label{HomoThm}
Let $K,K'$ be semirings, $h : K \rightarrow K'$, $Q$ an \SPJUlastA query and let $M$ be a commutative monoid.
 For every $(M,K)$-database $D$, $Q(\hRel{h}(D)) = \hRel{h}(Q(D))$
if and only if $h$ is a semiring homomorphism.
\end{theorem}

The proof is by induction on the query structure, and is
straightforward given that for the constructs of SPJU queries
homomorphism commutation was shown in \cite{GKT-PODS07}, while
commutation for the new $\Agg_M$ construct follows directly from the
definition.
\begin{example}
\label{QueryResult}

Consider the following $\NX$-relation $R$: %(ignore for now the rightmost column).

$$
\begin{array}{|c||l|}
\hline
 \mathit{Sal} & \\
\hline
\hline
20 & r_1 \\
10 & r_2 \\
30 & r_3 \\
\hline
\end{array}
$$
%
%{\eat
%$$
%\begin{array}{|cc||l|}
%%\multicolumn{4}{c}{R} \\
%%\multicolumn{4}{c}{} \\
%\hline
%\mathit{Department} & \mathit{SumSalary} & \\
%\hline
%\hline
%\mathsf{d_1} & r_1 \otimes 20 +_{\NX \otimes {\cal N}} r_2 \otimes 10 & \delta(r_1+_{\NX} r_2) \\
%\mathsf{d_2} & r_3 \otimes 10 & r_3 \\
%\hline
%\end{array}
%$$
%}

Let $M$ be some commutative monoid, then $\Agg_M(R)$ consist of a single tuple with value $r_1\notimes 20+_{\sKM}r_2\notimes 10 +_{\sKM}r_3\notimes 30$. The intuition is that this
value captures multiple possible aggregation values, each of which may be
obtained by mapping the $r_i$ annotations to $\Nat$, standing for the multiplicity of the corresponding tuple. The commutation with homomorphism allows us to
first evaluate the query and only then map the $r_i$'s, changing directly the expression in the query result. For example, if $M=SUM$ and we map $r_1$ to 1,$r_2$ to $0$,$r_3$ to $2$, we obtain $1 \notimes 20 +_{\sKM} 2 \notimes 30= 1 \notimes 20 +_{\sKM} 1 \notimes 30+_{\sKM} 1 \notimes 30= 1\otimes 80$ (which
corresponds to the $M$ element $80$). As another example, the commutation with homomorphisms
allows us to propagate the deletion of the first tuple in $R$, by simply setting in the aggregation result $r_1=0$ (keeping the other annotations intact) and obtaining $2 \otimes 30= (1+1) \otimes 30$ =
$1 \otimes 30+ 1 \otimes30 = 1\otimes (30+30) = 1 \otimes 60$.
\end{example}

We further demonstrate an application for {\em security}.

\begin{example}
Consider the following relation $R$, annotated by elements from the security semiring $\bbS$.
$$
\begin{array}{|c||l|}
\hline
 \mathit{Sal} & \\
\hline
\hline
20 & \mathsf{S} \\
10 & 1_{\sSS} \\
30 & \mathsf{S} \\
\hline
\end{array}
$$

Recall (from Section \ref{semiringsubsec}) the order relation $1_{\sSS}<\mathsf{C}<\mathsf{S}<\mathsf{T}<0_{\sSS}$;
a user with credentials $\mathsf{cred}$ can only view tuples annotated with security level equal or less than $\mathsf{cred}$. %Recall that   Consequently, the following
%are axioms in $ \bbS$:
%
%\begin{itemize}
%\item $\forall c_1 \leq c_2. c_1+_{\bbS}c_2=c_1$ (to see a tuple with two alternative computations a user must be able to see one of them).
%\item $\forall c_1 \leq c_2. c_1*_{\bbS}c_2=c_2$ (to see a tuple with a joint computation a user must be able to see both parts).
%\end{itemize}
Now let $M = \MaxAgg$ and we obtain:
$\Agg_{\sMaxAgg}(R)= \mathsf{S}\notimes 20+_{\sKM}1_{\sSS}\notimes 10 +_{\sKM}\mathsf{S}\notimes 30 = \mathsf{S} \notimes (20 +_{MAX} 30) +
1_{\sSS} \notimes 10$ and we get $\mathsf{S} \notimes 30 + 1_{\sSS} \notimes 10$.

Assume now that we wish to compute the query results as viewed by a user with security credentials $\mathsf{cred}$. A naive computation would
delete from $R$ all tuples that require higher credentials, and re-evaluate the query (which in general may be complex). But observe that
the deletion of tuples is equivalent to applying to $R$ a homomorphism that
maps every annotation $t>\mathsf{cred}$ to $0$, and $t\leq \mathsf{cred}$ to $1$. Using homomorphism commutation we can do better
by applying this homomorphism only on the result representation (namely $\mathsf{S} \notimes 30 + 1_{\sSS} \notimes 10$).
For example, for a user with credentials $\mathsf{C}$, we map $\mathsf{S}$ to $0$  and $1_{\sSS}$ to $1$, and obtain
$0 \notimes 30 + 1 \notimes 10 = 1 \otimes 10$; similarly for a user with credentials $\mathsf{S}$ we get $1 \notimes 30 + 1 \notimes 10
= 1 \notimes (30+_{\sMaxAgg} 10)= 1 \otimes 30$.

\end{example}

%$$
%\begin{array}{|cc||l|}
%%\multicolumn{3}{c}{T} \\
%%\multicolumn{3}{c}{} \\
%\hline
%\mathit{Rank} & \mathit{SumSalary} & \\
%\hline
%\hline
%\mathsf{a} & S\otimes20+TS\otimes10 & \delta(S+TS)\\
%\mathsf{b} & (S+S)\otimes10 & \delta(S+S) \\
%\mathsf{c} & (P \otimes 15 + P \otimes 10) & \delta(P+P) \\
%\hline
%\end{array}
%$$

From the above definition of the semantics for aggregation, it is obvious
that the poly-size overhead property is fulfilled. Indeed, consider the
case of provenance for summation as in Example \ref{QueryResult}, and compare it to the naive representation provided in the Introduction. Instead of
having to list all (exponentially many) options for the sum of salaries,
we used an expression in $K \otimes \SumAgg$ that is of linear
size with respect to the input to the aggregation. As exemplified, the possible aggregate answers
now correspond to different valuations for the provenance tokens,
applied to this expression.

\subsection{Group By}

So far we have considered aggregation in a limited context, where the input relation contains a single attribute. In common cases, however, aggregation is used on arbitrary relations and
in conjunction with grouping, so we next extend the algebra to handle such an operation. The idea behind the construction is quite simple: we separately group the tuples according to
the values of their ``group-by" attributes, and the aggregated values for each such group are
computed similarly to the computation for the $AGG$ operator. When considering the {\em annotation} of the aggregated tuple, we encounter a technical difficulty: we want this annotation to be equal
$1_{\sK}$ if the input relation includes at least one tuple in the corresponding group, and $0_{\sK}$ otherwise (for intuition, consider the case of bag relations, in which the aggregated result can have at most
multiplicity 1); we consequently enrich our structure to include an additional construct $\delta$ that will capture that,
as follows:

%
%but of course, we cannot just use a notation of $0_{\sK}$ or $1_{\sK}$, as we need to account
%for e.g. deletion propagation or different security views.

\begin{definition}
A (commutative) $\delta${-semiring} is an algebraic structure
$(K,+_{\sK},\cdot_{\sK} ,0_{\sK} ,1_{\sK},\delta_{\sK})$ where
$(K,+_{\sK},\cdot_{\sK} ,0_{\sK},1_{\sK})$ is a commutative semiring
and $\delta_{\sK} :K \rightarrow K$ is a unary operation satisfying
the ``$\delta$-laws''
$\delta\sK (0_{\sK}) = 0_{\sK}$ and
$\delta_{\sK} (n1_{\sK}) = 1_{\sK} $ for all $n\geq 1$.
%where $n1_{\sK}  = 1_{\sK} +\cdots+_{\sK}  1_{\sK} $ ($n$ times).
If $K$ and $K'$ are $\delta$-semirings then a {\em homomorphism}
between them is a semiring homomorphism $h:K\rightarrow K'$,
for which we have in addition $h(\delta_{\sK}(k)) = \delta_{\sKprime}(h(k))$.
\end{definition}

The $\delta$-laws completely determine $\delta_{\bbB}$ and
$\delta_{\bbN}$. But they leave a lot of freedom for the
definition of $\delta_{\sK}$ in other semirings;
in particular for the security semiring,
a reasonable choice for $\delta_{\bbS}$ is the identity function.
%
%But they leave a lot of freedom for the
%definition of $\delta$ in other semirings \val{we must discuss
%it for $\bbS$, and the bag-security and $\NX$, etc. Please remind
%me why I cannot set $\delta_{\bbS}(\mathsf{S})=1_{\bbS}$ ? }

As with any equational axiomatization, we can construct the
commutative $\delta$-semiring {\em freely generated}
by a set $X$, denoted $\NXD$, by taking the quotient of the set of \\
$\{+,\cdot ,0 ,1 ,\delta \}$-algebraic expressions
by the congruence generated by the equational laws of commutative semirings
and the $\delta$-laws. For example, if $e$ and $e'$ are elements of
$\NXD$ (i.e., congruence classes of expressions given by some representatives)
then $e+_{\NXD}e'$ is the congruence class of the expression $e+e'$.
The elements of $\NXD$ are not standard polynomials
but certain subexpressions can be put in polynomial form,
for example $\delta(2+3xy^2)$ or $3+2\delta(x^2+2y)z^2$.
%where we have introduce parenthesis to disambiguate as needed.

%$\delta_{\NXD}(3+2\delta(x^2+2y)z^2) = \delta(3+2\delta(x^2+2y)z^2)$
%and $\delta_{\NXD}(3) = \delta(3) = \delta(1+1+1)=1$
%
%
%A semiring homomorphism $h:K \rightarrow K'$ can be extended to $\delta$-semirings by defining $h(\delta_{\sK}(k))=\delta_{K'}(h(k))$.

We are now ready to define the group by (denoted $GB$) operation; subsequently we exemplify its use, including in particular the role of $\delta$:

\begin{definition}
\label{DefAgg}
Let $R$ be a $K$-relation on set of attributes $U$, let $U' \subseteq U$ be a subset of attributes that will be grouped
and $U'' \in U$ be the subset of attributes with values in $M$ (to be aggregated). We assume that $U'\cap U''=\emptyset$. For a tuple $t$, we define $T = \{t' \in supp(R) \mid \forall u \in U'~~
t'(u)=t(u)\}$.

We then define the aggregation result $R'=GB_{U',U''}(R)$ as follows:

\begin{equation*}
R'(t)=
\begin{cases}
\text{$\delta_{\sK}(\Sigma_{t' \in T}R(t'))$} & \text{$T \neq \phi$, and} \\
\text{} & \text{$\forall u \in U''~~ t(u)=\Sigma_{t' \in T} R(t') \otimes t'(u)$}
\\
0 &\text{Otherwise.}
\end{cases}
\end{equation*}

\end{definition}

%denote by $InGroup(t,R,U')$ the set \\$\{t'\in supp(R) | \pi_{U'}(t)=\pi_{U'}(t') \}$. Further denote
%$InAgg(t,R,U',U'') = InGroup(t,R,U')$
%if $\forall u\in U'', \pi_u(t)=\Sigma_{t''\in InGroup(t,R,U')}R(t'')\otimes t''(u)$, and otherwise $InAgg(t,R,U',U'') = \phi$. \\
%Finally, we define $GB_{U',U''}R: (U,K \otimes M)\!-\!tuples \rightarrow K$  by
%$(GB_{U',U''}R)(t) = \delta(\Sigma_{t' \in InAgg(t,R,U',U'')}R(t'))$.
%%Given a tuple $t$ and a subset $U$ of its attributes, we use $t |_{U}$ to denote the restriction of $t$ to the attributes $U$.
%%If $R:(U,K \otimes M)\!-\!tuples \rightarrow K$, $U_1 \subseteq U$ is a subset of attributes and
%%$u \in U$ is the aggregated attribute, then $Group\!-\!By_{U_1,u}R: (U,K \otimes M)\!-\!tuples \rightarrow K$ is defined by
%%$[Group\!-\!By_{U_1,u_2}R](t) =\Sigma_{t_i \in supp(R)}R(t_i)*\prod_{u_1 \in U_1}[t_i|_{u_1} = t_j|_{u_1}] $ where $t_j \in supp(R)$ is such that $\forall u_1 \in U_1. t|_{u_1} = t_j|_{u_1}$ and $t|_{u_2} = \sum_{t_{k}} R(t_{k})* \prod_{u_1 \in U_1} [t_{k}|_{u_1} = t_{j}|_{u_1}] \otimes t_{k}|_{u_2}$
%\end{definition}

%We next exemplify the result of an aggregation query according to definition \ref{DefAgg}, and specifically illustrate the use of $\delta$.

\begin{example}
\label{AlgebraExample}

Consider the relation $R$:
$$
\begin{array}{|cc||l|}
\hline
\mathit{Dept} & \mathit{Sal} & \\
\hline
\hline
\mathsf{d_1} &  20 & r_1 \\
\mathsf{d_1} & 10 & r_2 \\
\mathsf{d_2} & 10 & r_3 \\
\hline
\end{array}
$$
and a query $GB_{\{Dept\},Sal}R$, where the monoid used is SUM. The result (denoted R') is:
$$
\begin{array}{|cc||l|}
%\multicolumn{4}{c}{R} \\
%\multicolumn{4}{c}{} \\
\hline
\mathit{Dept} & \mathit{Sal} & \\
\hline
\hline
\mathsf{d_1} & r_1 \notimes 20 +_{\sK \notimes SUM} r_2 \notimes 10 & \delta\sK (r_1+_{\sK}  r_2) \\
\mathsf{d_2} & r_3 \notimes 10 & \delta\sK (r_3) \\
\hline
\end{array}
$$
Each aggregated value (for each department) is computed very similarly to the computation in Example \ref{QueryResult}. Consider the provenance annotation of the first tuple: intuitively,
we expect it to be $1_{\sK}$ if at least one of the first two tuples of $R$ exists, i.e. if at least one out of $r_1$ or $r_2$ is non-zero. Indeed the expression is $\delta(r_1+_{\sK}  r_2)$
and if we map $r_1,r_2$ to e.g. $2$ and $1$ respectively, we obtain $\delta_{\Nat}(3) =1$. %Similarly if we use annotations from the security semiring we could get expressions of the sort
%$\delta(S+T) = \delta(T)$ which will be interpreted as 1 (0) for users with credentials $c\leq T$ ($c>T$).
\end{example}

We use \SPJUlastAGB as the name for relational algebra with the two new operators AGG and GB. We note that the poly-size overhead property is still
fulfilled for queries in \SPJUlastAGB; commutation with homomorphism also extends to \SPJUlastAGB (see proof in the Appendix).

Recall that an additional desideratum from the semantics was bag / set compatibility. Recall that sets and bags are modeled by $K=\bbN$ and $K=\bbB$ respectively. We next study compatibility in a more general way,
for arbitrary $K$ and $M$. %\val{This is surely a matter of taste, but

\subsection{Annotation-aggregation compatibility}
\label{Compat}
The first desideratum we listed was an obvious sanity check: whatever
semantics we define, when specialized to the familiar aggregates of
max, min and summation, it should produce familiar results. Since we
had to take an excursion through the tensor product $K\otimes M$, this
familiarity is not immediate. However, the following proposition holds (its correctness will follow from theorems \ref{compatibilityIdempotence} and
\ref{compatibilityHomomorphism}).
\begin{proposition}
\label{familiarity}
In the following constructions: $\bbB\otimes\MaxAgg$,
$\bbB\otimes\MinAgg$, and $\bbN\otimes\SumAgg$,
$\iota : M\rightarrow K\otimes M$ where
$\iota(m) = 1_{\sK} \otimes m$ is a monoid {\em isomorphism}.
\end{proposition}
and this means our semantics satisfies the set/bag compatibility property
because in these cases computing in $K\otimes M$ exactly mirrors
computing in $M$.

But of course, we are also interested in working with other semirings,
in particular the provenance semiring, for which $\NX \otimes M$ and $M$ are in general not isomorphic (in particular, $\iota$ is not
surjective and thus not an isomorphism). In fact, the whole point of working in $\NX\otimes\MaxAgg$,
for example, is to add annotated aggregate
computations to the domain of values. Most of these do not correspond to actual real numbers
as e.g. $\iota(\MaxAgg)$ is a strict subset of $\NX\otimes\MaxAgg$
(and similarly $\iota(\SumAgg)$ is a strict subset of $\NX\otimes\SumAgg$ etc.). However, when provenance tokens are
valuated to obtain set (or bag) instances, we can go back into
$\iota(\MaxAgg)$ (or $\iota(\SumAgg)$ etc.), and then we should obtain
familiar results by ``stripping off'' the $\iota$. It turns out that
this works correctly with $\NX$ but not necessarily with arbitrary
commutative semirings $K$. The reason is that not only that $\iota$ is not an isomorphism, but in general it may be
be {\em unfaithful} (not injective).
Indeed
$\iota :\SumAgg \rightarrow \bbB\otimes \SumAgg$ is not injective:
$$
\iota(4) = \iota(2+2) = \iota(2) +_{\sKM} \iota(2) =
\top \notimes 2 +_{\sKM} \top \notimes 2 =
$$
$$
= (\top \vee \top) \notimes 2 = \top \notimes 2 = \iota(2)
$$
This is not surprising, as it is related to the well-known
difficulty of making summation work properly with set semantics. In general, we thus define compatibility as follows:

\begin{definition}
We say that a  commutative semiring $K$ and a commutative monoid $M$ are
{\em compatible} if $\iota$ is injective. %an isomorphism between
%$M$ and $\beta_{\sK} \otimes M$.
\end{definition}

The point of the definition is that when there is compatibility,
we can work in $K\otimes M$ and whenever the results
belong to $\iota(M)$, we can {\em safely}
read them as familiar answers from $M$. We give three theorems that capture some general conditions for compatibility.

%\val{I'll split this theorem in two}

First, we note that if we work with a semiring in which $+_{\sK}$ is idempotent, such as $\bbB$ or $\bbS$, a compatible monoid must also be idempotent
(e.g. \MinAgg\ or \MaxAgg\ but not \SumAgg):
\begin{proposition}
Let $K$ be some commutative semiring such that $+_{\sK}$ is idempotent, and let $M$ be some commutative monoid. If $M$ is compatible with $K$, then $+\sM$ is idempotent.
\end{proposition}
\begin{proof}
$\iota(m) = 1_{\sK}\otimes m = (1_{\sK}+_{\sK}1_{\sK})\otimes m
= 1_{\sK}\otimes m +_{\sKM} 1_{\sK}\otimes m =
1_{\sK}\otimes (m +_{\sM} m) = \iota(m +_{\sM} m)$
\end{proof}

Nicely enough, idempotent aggregations are compatible with {\em every}
annotation semiring $K$ that is {\em positive with respect to $+_{\sK}$}. $K$ is said to be positive with respect to $+_{\sK}$ if $k+_{\sK}k'=0_{\sK} \Rightarrow k=k'=0_{\sK}$. For instance, $\bbB$, $\bbN$, $\bbS$ and $\NX$ are such semirings (but not $(\bbZ,+,\cdot,0,1)$). The following theorem holds:

\begin{theorem}
\label{compatibilityIdempotence}
If $M$ is a commutative monoid such that $+_{\sM}$ is idempotent, then $M$ is compatible with any commutative semiring $K$ which is positive with respect to $+_{\sK}$.
\end{theorem}
\begin{proof} [sketch]
We define $h: K \notimes M \rightarrow M$ as \\* $h(\sum_{i \in I} k_{i} \notimes m_i) = \sum_{j \in J} m_j$ where
$J=\{j \in I | k_{j} \neq 0\}$. We can show that $h$
is well-defined (details deferred to the Appendix); since
$\forall m \in M~~ h \circ \iota(m) = m$, $\iota$ is injective and thus $K$ and $M$ are compatible.
\end{proof}

For general (and in particular non-idempotent) monoids (e.g. \SumAggWithSpace) we identify a sufficient condition on $K$ (which in particular holds for $\NX$), that allows for compatibility:

\begin{theorem}
\label{compatibilityHomomorphism}
Let $K$ be a commutative semiring.
If there exists a semiring homomorphism from $K$ to $\bbN$ then
$K$ is compatible with all commutative monoids.
\end{theorem}

\begin{proof} [sketch]
Let $h'$ be a homomorphism from $K$ to $\bbN$, and $M$ be an arbitrary commutative monoid. We define a mapping $h:K\otimes M \rightarrow M$ by
$h\left(\Sigma k_i\otimes m_i\right) = \Sigma h'(k_i)m_i$.
We show in the Appendix that $h$ is well-defined and that
$h \circ \iota$ is the identity function hence $\iota$ is injective.
\end{proof}

\begin{corollary}
The semiring of provenance polynomials $\NX$ is
compatible with all commutative monoids.
\end{corollary}

\eat{
Note that there is no homomorphism from $\bbB$ ($\bbS$) to $\bbN$. The following lemma holds:

\begin{lemma}
The semirings $\bbB$ and $\bbS$ are compatible with a commutative monoid $M$ if and only if $+_{\sM}$ is idempotent.
\end{lemma}

The above properties are not surprising, in the sense that $\bbB$ and
$\bbS$-annotated databases are based on set semantics, while
$\bbN$-annotated databases capture bag semantics;
, since it is known that
some aggregate queries such as max and min are ``natural" for
set-semantics while other aggregates such as sum and product require
the use of bag semantics.  Aggregates of the first class are indeed
idempotent, while those of the second class are not.
}

Now consider the security semiring $\bbS$. It is idempotent, and
therefore not compatible with non-idempotent monoids such as
$\SumAgg$.  Still, we want to be able to use $\bbS$ and other
non-idempotent semirings, while allowing the evaluation of aggregation
queries with non-idempotent aggregates. This would work if
we could construct annotations that would allow us to
use Theorem~\ref{compatibilityHomomorphism}, in other words,
if we could {\em combine} annotations from $\bbS$, with multiplicity
annotations (i.e. annotations from $\bbN$).  We explain next the
construction of such a semiring $\bbS\bbN$ (for security-bag), and its
compatibility with any commutative monoid $M$ will follow from the
existence of a homomorphism $h:$\bbS\bbN$ \rightarrow \bbN$.

\paragraph*{Constructing a compatible semiring} We start with the semiring of polynomials $\bbN[\bbS]$, i.e. polynomials where instead of indeterminates(variables) we have members of $\bbS$,
and the coefficients are natural numbers. Already $\bbN[\bbS]$ is compatible with any commutative monoid $M$, as there exists a homomorphism $h:\bbN[\bbS] \rightarrow \bbN$; but if we work with $\bbN[\bbS]$
%we lose the ability to use all axioms from $\bbS$ and
%in particular may obtain large (yet still of polynomial size)
%annotations in queries results.
we lose the ability to use the identities that hold in $\bbS$
and to thus reduce the size of annotations in query results.
We can do better by taking the
quotient of $\bbN[\bbS]$ by the smallest congruence containing
the following identities:
 \begin{itemize}
\item $\forall s_1,s_2 \in \bbS~~ s_1\geq s_2 \Longrightarrow s_1\cdot_{\bbN[\bbS]} s_2=s_1$.
\item $\forall c \in \bbN,s \in \bbS ~~ 0\cdot_{\bbN[\bbS]}s = c\cdot_{\bbN[\bbS]}0_{\bbS} = 0$.
\item $\forall c \in \bbN~~ c\cdot_{\bbN[\bbS]}1_{\bbS} = c$.
\end{itemize}
%\val{How about $0_{\bbS}$?}
We will denote the resulting quotient
semiring by $\bbS\bbN$.
%with its zero (one) element being $0$ ($1$ resp.).
It is easy to check that the faithfulness of the embeddings of
$\bbN$ and $\bbS$ in $\bbN[\bbS]$ is preserved
by taking the quotient. Most importantly,
$\bbS\bbN$ is still homomorphic to $\Nat$. Thus,
\begin{corollary}
$\bbS\bbN$ is compatible with any commutative monoid $M$.
\end{corollary}

\begin{example}

Consider the SUM monoid. Let $R,S$ be the following
$\bbS$-relations which by the embedding of $\bbS$ we
take as $\bbS\bbN$-relations:
$$
\begin{array}{|c||l|}
%\multicolumn{3}{c}{R} \\
%\multicolumn{3}{c}{} \\
\hline
\mathit{A} &  \\
\hline
\hline
\mathsf{30} & \mathsf{S}\\
\hline
\multicolumn{2}{c}{} \\
\multicolumn{2}{c}{R}
\end{array}
~~~
\begin{array}{|c||l|}
%\multicolumn{3}{c}{T} \\
%\multicolumn{3}{c}{} \\
\hline
\mathit{A} & \\
\hline
\hline
\mathsf{30} & \mathsf{T}\\
\mathsf{10} & 1_{\bbS}\\
\hline
\multicolumn{2}{c}{} \\
\multicolumn{2}{c}{S}
\end{array}
$$
Consider the query: $AGG(R \cup \Pi_{S.A}(S \bowtie R))$. Ignoring the annotations, the expected result (under bag semantics) is 70.
Working within the (compatible) semantics defined by
$\bbS\bbN\otimes\SumAgg$, the query result contains an aggregated value
of $(\mathsf{T} \cdot_{\bbS\bbN} \mathsf{S} +_{\bbS\bbN} \mathsf{S}) \otimes 30$
$+ \mathsf{S} \otimes 10$. We can further simplify this to $\mathsf{T}  \otimes 30 + \mathsf{S} \otimes 30 +
\mathsf{S} \otimes 10 = \mathsf{T}  \otimes 30 + \mathsf{S} \otimes 40$. This means that
e.g. for a user with credentials $\mathsf{T}$ the query result is $1_{\bbS\bbN} \otimes 70$, and we can use
the inverse of $\iota$ to map it to $\Nat$ and obtain $70$.  Similarly, for a user with
credentials $\mathsf{S}$, the query result is mapped to $40$. These are indeed the
expected results.

\end{example}

Note that if we would have used in the above example
$\bbS$ instead of $\bbS\bbN$ we would have
$(\mathsf{T} +_{\bbS} \mathsf{S}) = \mathsf{S}$ so
$(\mathsf{T} +_{\bbS} \mathsf{S}) \otimes 30$ would be the same as
$\mathsf{S} \otimes 30$. For a user with credentials $\mathsf{T}$ we could either use this, leading to the result of
$1_{\bbS} \otimes 40$, or use the same computation done in the example, to obtain $1_{\bbS} \otimes 70$. Indeed, in
$\mathsf{\bbS} \otimes SUM$, we have $1_{\bbS} \otimes 40 = 1_{\bbS} \otimes 70$. This is the same phenomenon demonstrated
in the beginning of this subsection for $\bbB$, where $\iota$ is not injective, preventing us from stripping it away.

Note also that if we would have used $\bbN[\bbS]$
instead of $\bbS\bbN$ then we could not
have done the illustrated simplifications.

\section{Nested Aggregation Queries}
\label{GeneralAggSec}
%\daniel{Should we be a bit more careful? E.g. queries that mix different agg. functions are not supported..}
So far we have studied only queries where the aggregation operator is the last one performed. In this section we extend the discussion to queries that involve comparisons on aggregate values. We first demonstrate the difficulties that arise in designing an algebra for such queries, then explain how to extend the construction to overcome these difficulties.

%in designing a proper algebra for general aggregation queries. Finally we explain how to overcome these difficulties by using a novel construction, applied uniformly to any given semiring.
%
%that the algebra suggested in the previous section does not commute with homomorphism; we then show that this is not a pitfall of the specific algebra suggested, but rather no algebra on $\NX \bigotimes M$ -valued relations can be a bag (or set) algebra and also commute with homomorphism. We then suggest an extended construction, and show that it does allow to achieve the desired algebra.
%

\paragraph*{Note} For simplicity, all results and examples are presented for queries in which the comparison operator is equality ($=$). However the results can easily be extended to arbitrary comparison predicates, that can be decided for elements of $M$.

\subsection{Difficulties}

We start by exemplifying where the algebra proposed for restricted aggregation queries, fails here:
\begin{example}
\label{ImpossibleExample}

Reconsider the relation (denoted $R'$) which is the result of aggregation query, depicted in Example \ref{AlgebraExample}. Further consider a query $Q_{select}$ that selects from $R'$ all tuples for which the aggregated salary equals $20$. The crux is that deciding the truth value of the selection condition involves interpreting the comparison operator on symbolic representation of values in $R'$; so far, we have no way of interpreting the obtained comparison expression, for instance $r_1 \otimes 20 + r_2 \otimes 10$ ``equals" $20$, and thus we cannot decide the existence of tuples in the selection result.
\end{example}

%
%
%The trivial extension of the algebra presented in the
%previous section would
%
%
%
%
%
%
%
%Let $R$ be the following relation:
%
%
%And let $S$ be the following relation:
%$$
%\begin{array}{|cc||l|}
%%\multicolumn{4}{c}{R} \\
%%\multicolumn{4}{c}{} \\
%\hline
%\mathit{D} & \mathit{E} & \\
%\hline
%\hline
%1 & 200 & t_3\\
%2 & 100 & t_4\\
%\hline
%\end{array}
%$$
%
%And consider the following query $Q$:
%
%\begin{verbatim}
%(SELECT A, SUM(B) As SumB
%FROM R
%GROUP BY A)
%INNER JOIN R ON
%SumB = D
%\end{verbatim}
%
%Evaluating $Q(D)$ according to the proposed algebra starts by evaluating the aggregation query which results in a relation with
%a single tuple whose $SumB$ value is $t_1 \otimes 1 + t_2 \otimes 1$. Then we join the query results with $R$, and the joint tuples should appear in the result based on the condition  $t_1 \otimes 1 + t_2 \otimes 1 = 1 (=2)$.
%
%% which is incomparable to 1 or 2. But upon joining the tables, there will be a tuple in the joined table only upon equality. The natural extension of the algebra to handle such cases would extend the condition on joined values to account to equality in $K \otimes M$; but using this definition, the result of the join in this example will be empty, and it will stay empty after applying any homomorphism on it. This contradicts commutation with homomorphism: if for example we apply on $R$ a homomorphism that sets $t_1=t_2=t_3=t_4=1$ , then the query results contains two tuples.

%\end{example}

Note that in the above example, the truth value of the comparison (and consequently the set of tuples in the query result) depends in a {\em non-monotonic way} on the existence of tuples in the (original) input relation $R$: note that if we map $r_1$ to $1$ and $r_2$ to $0$ then the tuple with dept. $d_1$ appears in the query result, but if we map both to $1$, it does not. The challenge that this non-monotonicity poses is fundamental, and is encountered by any algebra on
$(M,K)$-relations. The following proposition, which is the counterpart of proposition \ref{noMK}, holds (proof deferred to the Appendix):

\begin{proposition}
\label{noMKGenAgg}

There is no $(M,K)$-relation semantics for nested aggregation queries with \MaxAgg-(or \MinAgg-)-aggregation
that is both set-compatible and commutes with homomorphisms. Similarly for \SumAgg-aggregation and bag-compatibility.
%there is no $(M,K)$-relation semantics for nested aggregation queries with
%that is both bag-compatible and commutes with homomorphisms.
\end{proposition}
\vspace{-2mm}

Consequently, a more intricate construction is required for nested aggregation queries. %We next explain this construction.

\subsection{An Extended Structure}

We start with an example of our treatment of nested aggregation queries, then give the formal construction.

\begin{example}
\label{ConstructionExample}
Reconsider example \ref{ImpossibleExample}, and recall that the challenge
in query evaluation lies in comparing elements of $K \notimes M$ with elements of $M$ (or $K \notimes M$, e.g. in case of joins). Our solution is to introduce to the semiring $K$ new elements, of the form $[x=y]$ where $x,y \in K \otimes M$ (if we need to compare with $m \in M$, we use $\iota(m)$ instead). The result of evaluating the query in example \ref{ImpossibleExample} (using $M= \SumAgg$) will then be captured by: %then has the following form:

$$
\begin{array}{|cl|l|}
%\hspace{-10mm}
%\multicolumn{4}{c}{R} \\
%\multicolumn{4}{c}{} \\
\hline
\mathit{Dept} & \mathit{Sal} & \\ %& \mathit{D} & \mathit{E} & \\
\hline
\hline
\mathsf{d}_1 & r_1 \otimes 20  &  \delta(r_1+_{\sK} r_2)\!\cdot_{\sK} \\
& +_{\sKM} r_2 \otimes 10 & \left[r_1 \otimes 20 +_{\sKM} r_2 \otimes 10 = 1_{\sK}\otimes 20\right] \\[1mm]
\mathsf{d}_2 & r_3 \otimes 10 & \delta(r_3)\!\cdot_{\sK}\!\left[r_3 \otimes 10=1_{\sK}\otimes 20\right]\\
\hline
\end{array}
$$

Intuitively, since we do not know which tuples will satisfy the selection criterion, we keep both tuples
and multiply the provenance annotation of each of them by a symbolic equality expression. These equality expressions are kept as symbols until
we can embed the values in $M=\SumAgg$ and decide the equality (e.g. if $K=\Nat$), in which case we ``replace" it by $1_{\sK}$ if it holds or $0_{\sK}$ otherwise. For example, given a homomorphism $h: \NX \rightarrow \bbN$, $h(r_1)=h(r_2)=1$, then $h^M(r_1 \notimes 20 +\sKM r_2 \notimes 10) = h(r_1)\otimes 20 +\sKM h(r_2)\otimes 10 = 1\otimes 30 \neq 1\otimes 20 $, thus the equality expression is replaced  with (i.e. mapped by the homomorphism to) $0_{\sK}$.
\end{example}

We next define the construction formally; the idea underlying the construction is to define a semiring whose elements are polynomials, in which equation elements are additional indeterminates. To achieve that, we introduce for any semiring $K$ and any commutative monoid $M$, the ``domain'' equation
$\widehat{K} = \Nat[K \cup
\{[c_1=c_2] \mid c_1,c_2\in \widehat{K}\notimes M\}]$.
The right-hand-side is a monotone, in fact continuous w.r.t. the usual
set inclusion operator, hence this equation has a set-theoretic least
solution (no need for order-theoretic domain theory).
The solution also has an obvious commutative semiring structure
induced by that of polynomials. The solution
semiring is $\widehat{K}=(X,+_{\widehat{K}},\cdot_{\widehat{K}},0_{\widehat{K}},1_{\widehat{K}})$, and we continue by taking the quotient on $\widehat{K}$
defined by the following axioms.

For all $k_1,k_2 \in K, c_1,c_2,c_3,c_4 \in \widehat{K}\otimes M$:
%\vspace{-1mm}
\begin{eqnarray*}
                   0\!_{\widehat{K}} & \sim & 0_{\sK} \\
                   1\!_{\widehat{K}} & \sim & 1_{\sK} \\
                                 k_1+\!_{\widehat{K}}k_2 & \sim & k_1+_{\sK} k_2  \\
                           k_1\!\cdot\!_{\widehat{K}}k_2 & \sim & k_1+_{\sK} k_2  \\
  \left[c_1 =c_3\right] & \sim &  \left[c_2=c_4\right] \text{(if $c_1 =\!_{\widehat{K} \otimes M}c_2$, $c_3 =\!_{\widehat{K} \otimes M}c_4$)}\\
\end{eqnarray*}
%\vspace{-5mm}
%\begin{enumerate}
%\item $0_{\hat{K}}=_{\hat{K}}0_{\sK},1_{\hat{K}}=_{\hat{K}}1_{\sK}$
%\item $\forall a,b \in K. a+_{\hat{K}}b =_{\hat{K}} a+_{\sK}b$
%\item $\forall a,b \in K. a\cdot_{\hat{K}}b =_{\hat{K}} a\cdot_{\sK} b$
%\item $\forall c_1,c_2\in \hat{K}\notimes M. [c_1=c_2]=_{\hat{K}}[c_2=c_1]$
%\item $\forall c_1,c_2,c_3 \in \hat{K} \otimes M$. $c_1 =_{\hat{K} \otimes M}c_2 \Rightarrow [c_1 =c_3] =_{\hat{K}} [c_2=c_3]$,
%
%% $a =_{\sK}b \Rightarrow [a*m_1 =_{K^{M}}c*m_2] = [b*m_1 =_{K^{M}}c*m_2]$
%%if and only if $a =_{\sK} b$
%\end{enumerate}
and if $K$ and $M$ are such that $\iota$ defined by $\iota(m)=1_{\sK} \notimes m$ is an isomorphism (and let $h$ be its inverse), we further take the quotient defined by: for all $a,b \in K \otimes M$,
%\vspace{-1mm}
\begin{eqnarray*}
\hspace{-30mm}
\text{(*)} \left[a=b\right] & \sim & 1_{\sK}  \text{(if $h(a)=\sM h(b)$)} \\
           \left[a=b\right] & \sim & 0_{\sK}  \text{(if $h(a)\neq\sM h(b)$)} \\
\end{eqnarray*}
\vspace{-5mm}
%\begin{itemize}
%\item (*) $\forall a, b \in K \notimes M.h(a)=_{\sM} h(b) \Leftrightarrow [a=b] = 1_{\sK}$, $h(a)\neq_{\sM} h(b) \Leftrightarrow [a=b] = 0_{\sK}$
%\end{itemize}

%It is of course parameterized by the commutative monoid $M$ that is used.
%
%% and the resulting structure is denoted $K^{M}$.
%
%\begin{definition}
%\label{RecursiveDef}
%Given a semiring $K$ and a commutative monoid $M$, we  define  taking the quotient defined by the
%following congruence:
%
%\begin{enumerate}
%\item  $\forall a,b \in K$. $a =_{\sK}b \Leftrightarrow a =_{K^{M}} b$
%\item $\forall c_1,c_2\in K^{M}\notimes M. [c_1=c_2]=[c_2=c_1]$
%\item $\forall c_1,c_2,c_3 \in K^{m} \otimes M$. $c_1 =_{\sK}c_2 \Rightarrow [c_1 =_{K^{M}}c_3] = [c_2=c_3]$,
%
%% $a =_{\sK}b \Rightarrow [a*m_1 =_{K^{M}}c*m_2] = [b*m_1 =_{K^{M}}c*m_2]$
%%if and only if $a =_{\sK} b$
%\end{enumerate}
%
%
%%if and only if $$
%%\item If $a \in K \otimes M, b \in M$ then $[a=b] = 1_{\sK}$ if  $a=1 \otimes b$.
%%\end{enumerate}
%
%\end{definition}
%
%It is easy to observe that $K^{M}$ is a semiring.

We use $K^M$ to denote the semiring obtained by applying the above construction on a semiring $K$ and a commutative monoid $M$. A key property is that, when we are able to interpret the equalities in $M$, $K^M$ collapses to $K$. Formally,

\begin{proposition}
\label{SolveMK}
If $K$ and $M$ are such that $K \otimes M$ and $M$ are isomorphic via $\iota$ then $K^{M}=K$.
\end{proposition}

The proof (deferred to the Appendix) is by induction on the structure of elements in $K^{M}$, showing that at each step we can ``solve" an equality sub-expression, and replace it with $0_{\sK}$ or $1_{\sK}$.

\paragraph*{Lifting homomorphisms} To conclude the description of the construction we explain how to lift a semiring homomorphism from $h:K \rightarrow K'$ to $h^{M}:K^{M} \rightarrow K'^{M}$,
for any commutative monoid $M$ and semirings $K,K'$. $h^{M}$ is defined
 recursively on the structure of $a \in K^{M}$: if $a \in K$ we define  $h^{M}(a) = h(a)$, otherwise $a=\left[b \otimes m_1=c \otimes m_2\right]$ for some $b,c \in K^{M}$ and $m_1,m_2 \in M$ and we define $h^{M}(a) =
\left[h^{M}(b) \otimes m_1= h^{M}(c) \otimes m_2\right]$. Note that the application of a homomorphism $h^{M}$ maps equality expressions to equality expressions (in which elements in $K'$ appear instead of elements of $K$ appeared before). If $K'$ and $M$ are such that their corresponding $\iota: M \rightarrow K \otimes M$ defined by $\iota(m)=1_{\sK} \otimes M$  is injective, then we may ``resolve the equalities", otherwise the (new) equality expression remains.

\subsection{The Extended Semantics}

The extended semiring construction allows us to design a semantics for general aggregation queries. Intuitively, when the existence of a tuple in the result relies on the result of a comparison involving aggregate values (as in the result of applying selection or joins), we multiply the tuple annotation by the corresponding equation annotation.

 In the sequel we assume, to simplify the definition, that the query aggregates and compares only values of $K^M\otimes M$ (a value $m \in M$ is first replaced by $\iota(m)=1_{\sK}\otimes m$). In what follows, let $R (R_{1},R_{2})$ be $(M,K^{M})$-relations on an attributes set $U$. Recall that for a tuple $t$, $t(u)$ (where $u\in U$) is the value of the attribute $u$ in $t$; also for $U'\subseteq U$, recall that we use $t\mid_{U'}$ to denote the restriction of $t$ to the attributes in $U'$. Last, we use $(K^M\otimes M)^{U}$ to denote the set of all tuples on attributes set $U$, with values from $K^M \otimes M$. The semantics follows:

%Also, for a tuple $t$ and an attribute (or set of attributes) $u$, let $t|_{u}$ be the tuple obtained from $t$ by taking only its values in the $u$ attribute.

%We further assume that the comparisons made in the query are only on values from $K^M\otimes M$; comparison on $M$ or on $V$ can be easily decided and do not require the construction described below.

\eat
{
%(for simplicity of the definition, that we only make comparisons on values from $K^M\otimes M$, which is the complex case. However, it is generally possible that we would have to make comparisons of values from $M$ or even $V$. However, the solution for these cases is much simpler - where an equality expression on the values of $V$ is required by the algebra, we simply make the comparison and use $1_{\sK}$ or $0_{\sK}$ accordingly. Another issue is that we can now aggregate on values from either $M$ or $K^M\otimes M$. Again to simplify the definition we assume that we are only aggregating on values from $K^M\otimes M$; where an aggregation of values from $M$ is required, we can first replace each value $m\in M$ with $1_{\sK}\otimes m$, and then continue according to the algebra on $K^M\otimes M$.\yael{Put this here? perhaps move it into the algebra definition?}
}

%Consider the following algebra:

\vspace{-2mm}
\begin{enumerate}
\item {\em empty relation:} $\forall t~~ \phi(t) = 0$.
\item {\em union:}
%\hspace{-15mm}
%\begin{tabbing}
%\noindent \textbf{If} \\
%\indent $t \in supp(R_{1})\cup supp(R_{2})$ \\
%\noindent \textbf{Then} \\
%\indent $\left(R_1 \cup R_2\right)(t) =$ $\sum_{t' \in supp(R_{1})}R_{1}(t')\!\cdot_{\sK}$ $\prod_{u \in U}[t'(u)=t(u)]$ \\
%\noindent \textbf{Otherwise} \\
%\indent $\left(R_1 \cup R_2\right)(t) = 0 $ \\
%\end{tabbing}
$\left(R_1 \cup R_2\right)(t) =$
\begin{equation*}
\hspace{-15mm}
\begin{cases}
\text{$\sum_{t' \in supp(R_{1})}R_{1}(t')\cdot \prod_{u \in U}[t'(u)=t(u)]$} & \text{if $t \in supp(R_{1})$} \\
\text{$+ \sum_{t' \in supp(R_{2})}R_{2}(t') \cdot \prod_{u \in U}[t'(u)=t(u)]$}  & \text{$\cup supp(R_{2})$} \\[3mm]
\text{0} &\text{Otherwise.}
\end{cases}
\end{equation*}

\item {\em projection:} Let $U' \subseteq U$, and let $\mathrm{T} = \{t|_{U'} \mid t \in supp(R)\}$. Then 
$\Pi_{U'}(t) =$
\begin{equation*}
\hspace{-15mm}
\begin{cases}
\text{$\sum_{t' \in Supp(R)} R(t')\!\cdot\! \prod_{u\in U'}[t(u)=t'(u)]$} & \text{if $t \in \mathrm{T}$} \\[2mm]
\text{0} & \text{Otherwise.}
\end{cases}
\end{equation*}

\item {\em selection:} If %$R:(U,K^{M} \bigotimes M)-tuples \rightarrow K^{M}$ and
$P$ is an equality predicate involving the equation of some attribute $u\in U$ and a value $m\in M$
then $\left(\sigma_{P}(R)\right)(t)= R(t) \!\cdot\! \left[t(u) = \iota(m)\right]$.

\item {\em value based join:} We assume for simplicity that $R_1$ and $R_2$ have disjoint sets of attributes, $U_1$ and $U_2$ resp., and that the join is based on comparing a single attribute of each relation.
Let $u'_1\in U_1$ and $u'_2\in U_2$ be the attributes to join on. For every $t\in(K^M\otimes M)^{U_1\cup U_2}$:\\
\vspace{1mm}
    $\left(R_1\bowtie_{R_1.u_1=R_2.u_2}R_2\right)(t) = $\\ $R_1(t|_{U_1})\!\cdot\!R_2(t|_{U_2})\!\cdot_{\sK}\!\left[t(u_1)=t(u_2)\right]$.
\vspace{1mm}
\paragraph*{Simple Variants} Natural join (when $U_1$ and $U_2$ are not necessarily disjoint) is captured by a similar expression, with the equality sub-expression on the attributes common to $U_1$ and $U_2$; join on multiple values is captured by multiplication by the corresponding multiple equality expressions; in the representation of cartesian product (denoted by $\times$) no equality expressions appear (only $R_1(t|_{U_1})\!\cdot\!R_2(t|_{U_2})$).
%\begin{equation*}
%\hspace{-15mm}
%$R_1\bowtie_{R_1.u_1=R_2.u_2}R_2 = \\
%\begin{cases}
%\text{$\prod_{t' \in supp(R_{1})} R_{1}(t') \cdot [t'|_{u1} = t|_{u2}] \cdot$} & \text{$t|_{u1} \in supp(R_{1})$} \\
%\text{$ \prod_{t' \in supp(R_{2})} R_{2}(t') \cdot [t'|_{u1} = t|_{u2}]$} & \text {$\bigwedge t|_{u2} \in supp(R_{2})$} \\
%\text{0} & \text{otherwise}
%\end{cases}
%\end{equation*}
\item {\em Aggregation:}
$AGG_{\sM}(R)(t) =$
\begin{equation*}
\hspace{-15mm}
\begin{cases}
\text{ 1} & \text {$t(u) = \sum_{t' \in supp(R)} R(t')\!*_{\sKM}\! t'(u)$} \\[3mm]
\text{0} & \text{otherwise}
\end{cases}
\end{equation*}

\item {\em Group By:} Let $U' \subseteq U$ be a subset of attributes that will be grouped
and $u \in U \backslash U'$ be the aggregated attribute. Then for every $t\in(K^M\otimes M)^{U'\cup \{u\}}$:
%Denote $\prod_{U'}R$ by $R''$, and denote the aggregation result $GB_{U',U''}R$ by $R'$. Then:

$GB_{U',u}R(t) =$
\begin{equation*}
\hspace{-15mm}
\begin{cases}
\text{ $\delta(\left(\Pi_{U'}R\right)(t|_{U'}))$} & \text {$t(u) = \sum_{t' \in supp(R)} (R(t')\cdot_{\sK} $} \\
\text{} & \text{$\prod_{u\in U'}[t'(u) = t(u)])\!*_{\sKM}\! t'(u)$} \\[3mm]
\text{0} & \text{otherwise}
\end{cases}
\end{equation*}
\end{enumerate}

%
%If $R:(U,K^{M} \bigotimes M)-tuples \rightarrow K^{M}$, $U_1 \subseteq U$ is a subset of attributes and
%$u_2 \in U$ is the aggregated attribute, then $Group-By_{U_1,u_2}R: (U,K^{M} \bigotimes M)-tuples \rightarrow K^{M}$ is defined by
%$[Group-By_{U_1,u_2}R](t) =\Sigma_{t_i \in supp(R)}R(t_i)*\prod_{u_1 \in U_1}[t_i|_{u_1} = t_j|_{u_1}] $ where $t_j \in supp(R)$ is such that $\forall u_1 \in U_1. t|_{u_1} = t_j|_{u_1}$ and $t|_{u_2} = \sum_{t_{k}} R(t_{k})* \prod_{u_1 \in U_1} [t_{k}|_{u_1} = t_{j}|_{u_1}] \bigotimes t_{k}|_{u_2}$
%

%
%where the sum is taken over all $t_i$ tuples that agree with $t$ on the values of all attributes in $U_1$; the $u_2$ value of $t$ is $\Sigma_{K^{M} \bigotimes M}(R(t_i)\cdot_{K^{M} \bigotimes M}v_i)$ where the sum is taken over the same set of $t_i$ tuples, and $v_i$ is the value of $t_i$ in the $u_2$ attribute.
%\end{enumerate}

It is straightforward to show that the algebra satisfies set/bag compatibility and poly-size overhead; commutation with homomorphism is proved in the Appendix.

\begin{example}
Reconsider the relation in Example \ref{ConstructionExample}, and let us perform another sum aggregation on $\mathit{Sal}$. The value in the result now contains equation expressions:
\vspace{-2mm}
$$
\begin{array}{l}
\delta(r_1+\!_{\sK} r_2)\!\cdot\!_{\sK}\!\left[r_1 \otimes 20  +\!_{\sKM} r_2 \otimes 10 = 1_{\sK}\otimes 20\right] \\
\!*_{\sKM}\!\left(r_1 \otimes 20 +_{\sKM} r_2 \otimes 10\right)\\
+_{\sKM}\delta(r_3)\!\cdot_{\sK}\!\left[r_3 \otimes 10=1_{\sK}\otimes 20\right]\!*_{\sKM}\!r_3 \otimes 10\\
\end{array}
$$

Given a homomorphism $h:\NX\to \Nat$ we can ``solve" the equations, e.g. if $h(r_1)=1$, $h(r_2)=0$ and $h(r_3)=2$, we obtain an aggregated value of $1\otimes 40$. Note that the aggregation value is not monotone in $r_1,r_2,r_3$: map $r_2$ to $1$ (and keep $r_1$,$r_3$ as before), to obtain $1\otimes 20$.

\end{example}

\section{Difference}
\label{DifferenceSection}
We next show that via our semantics for aggregation, we can obtain for the first time a semantics for arbitrary queries with difference on $K$-relations. We describe the obtained semantics and study some of its properties.
%Nested aggregation queries allow to encode difference. Consequently we can capture, for the first time, arbitrary queries with difference on $K$-relations. We show the obtained semantics and study some of its properties.

\subsection{Semantics for Difference}
We first note that difference queries may be encoded as queries with aggregation, using the monoid $\widehat{\bbB}=(\{\bot,\top\},\vee,\bot)$ (the following encoding was inspired by \cite{K-PODS10,BNTW95}):

\begin{tabbing}
$R - S = \Pi_{a_1...a_n}\{\left(GB_{\{a_1,...a_n\},b} (R \times \bot_b \cup S \times \top_b)\right)$\\*~~~~~~~~~~~~ $\bowtie_{a_1,...a_n} (R \times \bot_b)\}$.
\end{tabbing}

%We have not described the algebra expression for the cartesian product $\times$, but it is exactly the same as that of the join operator, except for the equality expression which is omitted.
$\bot_b$ and $\top_b$ are relations on a single attribute $b$, containing a single tuple $(\bot)$ and $(\top)$ respectively, with provenance $1_{\sK}$. Using the semantics of Section \ref{GeneralAggSec}, we obtain a semantics for the difference operation.

Interestingly, we next show that the obtained semantics can be captured by a simple and intuitive expression. First, we note that since $\widehat{\bbB}$ is idempotent, every semiring $K$ positive with respect to $+_{\sK}$ is compatible with $\widehat{\bbB}$ (see Theorem \ref{compatibilityIdempotence}). The following proposition then holds for every $K,K'$ and every two $(\widehat{\bbB},K)$-relations $R,S$ (proof deferred to the Appendix):

\begin{proposition}
\label{Neg}
For every tuple $t$, semirings $K, K'$ such that $K'^{\widehat{\bbB}}\otimes \widehat{\bbB}$ is isomorphic to $\widehat{\bbB}$ via $\iota(m)=1_{\sK} \otimes m$, if $h:K\rightarrow K'$ is a semiring homomorphism then: \\*
$h^{\widehat{\bbB}}([(R-S)(t)]) = h^{\widehat{\bbB}}\left([S(t) \otimes \top = 0] \!\cdot_{\sK}\!R(t)\right)$.
\end{proposition}

The obtained provenance expression is thus ``equivalent" (in the precise sense of Proposition \ref{Neg}) to $[S(t) \otimes \top = 0] \cdot R(t)$. The following lemma helps us to understand the meaning of the obtained equality expression:

\begin{lemma}
For every semiring $K$ which is positive w.r.t. $+_{\sK}$ and $h:K\rightarrow \bbB$, $h^M\left(\left[S(t) \otimes \top = 0\right]\right) = \top$ iff $h(S(t)) = \bot$. %\yael{Because of the positivity of $\beta_{\sK}$,  we do not need the positivity of K!}
\end{lemma}
\begin{proof}
It is clear that if $h(S(t)) = \bot$, $h^M\left(\left[S(t) \otimes \top = 0\right]\right)= \left[h\left(S(t)\right) \otimes \top = 0\right] = \left[\bot \otimes \top = 0\right] = \top$. For the other direction, assume that $h(S(t)) = \top$. Thus $\left[h\left(S(t)\right) \otimes \top = 0\right] = \left[\top \otimes \top = 0\right]$. Since $\bbB$ and $\widehat{\bbB}$ are compatible, $\iota:\bbB\rightarrow \bbB\otimes \widehat{\bbB}$ is injective; thus $\iota(\bot) \neq \iota(\top)$; consequently $h^M\left(\left[S(t) \otimes \top = 0\right]\right) = \left[\top \otimes \top = 0\right] = \bot$. %\yael{Maybe relate the last part to the proof of the set algebra?}
\end{proof}

Consequently, the semantics can be interpreted as follows: a tuple $t$ appears in the result of $R-S$ if it appears in $R$, but does not appear in $S$. When the tuple appears in the result of $R-S$, it carries its original annotation from $R$. I.e. the existence of $t$ in $S$ is used as a boolean condition. %The following lemma proves that the above formulation expresses this semantics.

\begin{example}
\label{DiffEx}
Let $R,S$ be the following relations, where $R$ contains employees and their departments and $S$ containing departments that are designated to be closed:

$$
\begin{array}{|cc||l|}
%\multicolumn{4}{c}{R} \\
%\multicolumn{4}{c}{} \\
\hline
\mathit{ID} & \mathit{Dep} & \\
\hline
\hline
1 & \mathsf{d}_1 & t_1\\
2 & \mathsf{d}_1 & t_2 \\
2 & \mathsf{d}_2 & t_3\\
\hline
\multicolumn{3}{c}{} \\
\multicolumn{3}{c}{R}
\end{array}
~~~
\begin{array}{|c||l|}
%\multicolumn{4}{c}{R} \\
%\multicolumn{4}{c}{} \\
\hline
\mathit{Dep} & \\
\hline
\hline
\mathsf{d}_1 & t_4\\
\hline
\multicolumn{2}{c}{} \\
\multicolumn{2}{c}{S}
\end{array}
$$

To obtain a relation with all departments that remains active, we can use the query  $(\Pi_{Dep}R)-S$, resulting in:

$$
\begin{array}{|c||l|}
%\multicolumn{4}{c}{R} \\
%\multicolumn{4}{c}{} \\
\hline
\mathit{Dep} & \\
\hline
\hline
\mathsf{d}_1 & [t_4 \otimes \top = 0]\!\cdot\!(t_1+t_2) \\
\mathsf{d}_2 & [0=0]\!\cdot\!t_3\quad (=t_3)\\
\hline
\end{array}
$$

Now consider some homomorphism $h:\NX \rightarrow \bbN$ (multiplicity e.g. stands for number of employees in the department).
Note that if $h(t_4)>0$ then the department $\mathsf{d}_1$ is closed and indeed $\mathsf{d}_1$ is omitted from the support of the difference query result, otherwise it retains each original annotation that it had in $R$. Assume now that we decide to revoke the decision of closing the department $\mathsf{d}_1$. This corresponds to mapping $t_4$ to $0$; we can easily propagate this deletion to the query results; the equality appearing in the annotation of the first tuple is now $[0=0] = 1_{\sK}$  and we obtain as expected:

$$
\begin{array}{|c||l|}
%\multicolumn{4}{c}{R} \\
%\multicolumn{4}{c}{} \\
\hline
\mathit{Dep} & \\
\hline
\hline
\mathsf{d}_1 & t_1+t_2 \\
\mathsf{d}_2 & t_3\\
\hline
\end{array}
$$

\end{example}

%Combined with our definition of algebra for $RA^{+,aggr}$ semantics (Definition \ref{?}), we now have in hand a definition for a semantics relational algebra with aggregation) semantics.

In particular, we obtain a semantics for the entire Relational Algebra, including difference. It is interesting to study the specialization of the obtained semantics for particular semirings: $\bbB,\bbN,\bbZ$, and to compare it to previously studied semantics for difference.

\subsection{Comparison with other semantics}

For a semiring $K$ and a commutative monoid $M$  we say that two queries $Q,Q'$ are equivalent if for every input $(M,K)$-database $D$,
the results (including annotations) $Q(D)$ and $Q'(D)$ are congruent (namely the corresponding values and annotations are congruent) according to the axioms of $K^{M} \otimes M$ and $K^{M}$. In the sequel we fix $M = \widehat{B}$ and consider different instances of $K$, exemplifying different equivalence axioms for queries with difference while comparing them with previously suggested semantics. We use $Q \equiv_{\sK} Q'$ to denote the equivalence of $Q,Q'$ with respect to $K$ and $\widehat{B}$.

\paragraph*{$\bbB$-relations} For $K=\bbB$, our semantics is the same as set-semantics, thus the following proposition holds: %things are simple (the proof is immediate):

\begin{proposition}
For $Q,Q' \in RA$ it holds that $Q \equiv_{\sbbB} Q'$ if and only if $Q \equiv Q'$ under set semantics.
\end{proposition}
\vspace{-4mm}
\paragraph*{$\bbN$-relations} For $K=\bbN$, we compare our semantics to bag equivalence and observe that they are different (for queries with difference, even without aggregation). Intuitively this is because in our semantics, the righthand side of the difference is treated as a boolean condition, rather than having the effect of decreasing the multiplicity.  Formally,

\begin{proposition}
$Q \equiv_{\sbbN} Q'$ does not imply that $Q \equiv Q'$ under bag semantics, and vice versa.
\end{proposition}

\begin{proof}
Observe that $A - (B \cup B) \equiv_{\sbbN} A - B$; but this does not hold for bag semantics. In contrast, under bag semantics $(A \cup B) - B \equiv A$, but not for our semantics.
\end{proof}

\begin{example}
Reconsider Example \ref{DiffEx}, and let $t_1=t_2=t_3=t_4=1$. Under bag semantics, after projecting $R$ on the department attribute, the multiplicity of the department $\mathsf{d}_1$ becomes $2$; after applying the difference the department $\mathsf{d}_1$  is still in the result, but now with multiplicity $1$; in contrast under our semantics the department $\mathsf{d}_1$ does not appear in the support of the result.
\end{example}
\vspace{-3mm}
\paragraph*{$\bbZ$-relations} Finally, in \cite{GIT-ICDT09} the authors have presented $\bbZ$ semantics for difference, and have shown that it leads to equivalence axioms that are different from those that hold for queries with bag difference. It is also
different from the equivalence axioms that we have here for $\bbZ$-relations:

\begin{proposition}
$Q \equiv_{\sbbZ} Q'$ does not imply  $Q \equiv Q'$ under $\bbZ$ semantics \footnote{As defined in \cite{GIT-ICDT09}.}, and vice versa.
\end{proposition}

\begin{proof}
Under $\bbZ$ semantics it was shown in \cite{GIT-ICDT09} that $(A - (B - C)) \equiv (A \cup C) -  B$. This does not hold for our semantics. In contrast $A - (B \cup B) \equiv_{\sbbZ} A - B$, but this equivalence does not hold under $\bbZ$ semantics.
\end{proof}

\vspace{-5mm}

\paragraph*{Deciding Query Equivalence} We conclude with a note on the decidability of equivalence of queries using our semantics. It turns out that for semirings such as $\bbB,\bbN$ for which we can interpret the results in $\widehat{\bbB}$ (in the sense of proposition \ref{Neg} above), query equivalence is undecidable.
 %on provenance annotated relations. Given a semiring $K$, two queries are $K$-equivalent \cite{G-ICDT09} if their output is the same for every input $K$-relation. Equivalence of $RA$ queries under bag and set semantics is undecidable, but as shown in  \cite{G-ICDT09} undecidability of equivalence under these semantics do not necessarily entail undecidability for equivalence for queries on $\NX$-relations. However, under our proposed semantics for difference  undecidability of equivalence does hold:

\begin{proposition}
Let  $K$ be such that $K^{\widehat{\bbB}}\otimes \widehat{\bbB}$ is isomorphic to $\widehat{\bbB}$. Equivalence of Relational Algebra queries on $K$-relations is undecidable.
\end{proposition}

\begin{proof}
The proof is by reduction from equivalence under set semantics:  let $\phi$ be the empty query, i.e. a query whose answer always the empty relation. Given two $RA$ queries $Q,Q'$ (note that $Q$ and $Q'$ can include difference), their equivalence under set semantics holds if and only if $Q - Q' \equiv_{\sK} \phi$ and $Q' - Q\equiv_{\sK} \phi$.
\end{proof}

\section{Related Work}
\label{related}

Provenance information has been extensively studied in the database literature.
Different provenance management techniques are introduced
in~\cite{CuiWidomWiener00,ProvenanceBuneman,Why, Trio}, etc.,
and it was shown in \cite{GKT-PODS07,G-ICDT09} that
these approaches can be compared in the semiring framework.
To our knowledge, this work is the first to study aggregate queries in the context of provenance semirings. Provenance information has a variety of
applications (see introduction) and we believe that our novel framework for aggregate queries will benefit all of these. Specifically, queries with aggregation play a key role in modeling the operational logic of scientific {\em workflows} (see e.g. \cite{Susan,Workflows2}) and our framework is likely to facilitate a more fine-grained approach to workflow provenance.
\eat{
Workflow provenance has been so far managed mainly in a coarse-grained manner (i.e. at a flow level rather than  individual data items level).
Since aggregate queries play a key role in workflows, we believe that the framework presented here can be an important step towards fine-grained provenance management for workflows.
}

Aggregate queries have been extensively studied in
e.g. \cite{CohenSigRec,NuttCohenSagiv} for bag and set semantics. As explained in \cite{CohenSigRec}, such
queries are fundamental in many applications: OLAP queries, mobile
computing, the analysis of streaming data, etc.
We note that Monoids are used to capture general aggregation operators
in~\cite{NuttCohenSagiv}, but our paper seems to be the first to
study their interaction with provenance.
%We have discussed the specialization of the obtained algebra
%to set and bag semantics.
%as explained in \cite{CohenSigRec}, such queries are extensively used in
%The application of provenance to research in these areas is an interesting
%future research area.

Several semantics of {\em difference} on relations with annotations
have been proposed, starting with the $c$-tables of~\cite{IL84}.
The semirings with monus of~\cite{Userssemiring1} generalize this
as well as bag-semantics. Difference on relations with annotations
from $\bbZ$ are considered in~\cite{GIT-ICDT09} and from $\bbZ[X]$
in~\cite{G-Thesis}. As explained in Section 5, the semantics for difference defined
in this paper is different from all of these.

There are interesting connections between provenance management and query evaluation on uncertain (and probabilistic) databases (e.g. \cite{KO,suciu,Trio,KO2}), as observed in \cite{GKT-PODS07}. Evaluation of aggregate queries on probabilistic databases has been studied in e.g. \cite{ReHaving,DeshpandeVLDB09}. Trying to optimize the performance of aggregate query evaluation on probabilistic databases via provenance management is an intriguing future research challenge.

\vspace{-1mm}
\section{conclusion}
\label{ConcSection}

We have studied in this paper provenance information for queries with
{\em aggregation} in the semiring framework. We have identified three
desiderata for the assessment of candidate approaches: compatibility
with the usual set/bag semantics, commutation
with semiring homomorphisms and poly-size overhead. After showing
that approaches using provenance only to annotate the database tuples do
not satisfy all desiderata simultaneously, we considered a different
framework in which the computation of aggregate {\em values} is itself
annotated with provenance.  This has led us to the algebraic structure
of {\em semimodules} over commutative semirings of annotations and to
a tensor product construction for the semantics of annotated
aggregation.  The first product of this approach is a ``good" (i.e. satisfying the desiderata) semantics
for SPJU queries followed by an aggregation or a group-by with
aggregation. We have further studied the challenges that arise in evaluation of
queries that apply {\em comparisons} on aggregation results, e.g.,
joins over aggregate values, and shown that by careful adaptation of
the semimodule framework these challenges can be overcome with a semantics
that satisfies the desiderata. Finally,
we noted that {\em difference queries} may be encoded as queries with
aggregation, and studied the algebra induced for such queries.

We have exemplified in the paper the application of our approach for
deletion propagation and security annotations. As mentioned in the Introduction
and Related Work sections, there are various other areas in which
provenance is useful. Future research will focus on applying our
framework to the research tasks tackled in these areas.

%There has been recent advancement in the management of provenance information through {\em provenance graphs} that efficiently encode provenance polynomials; extending provenance graphs for aggregation queries. %An additional challenge lies in extending our algebra to other possible semantics of {\em difference}, such as bag difference.

\bibliographystyle{plain}
{\small
\bibliography{bib}
}
\newpage

\appendix

\section{SPJU algebra for K-relations}
\label{ProvenanceSemirings}

We recall the full definition of the SPJU algebra for K-relations from \cite{GKT-PODS07}:
\begin{description}
\item[Empty Relation] $\forall t~~ \phi(t) = 0$.
\item[Selection]  If $R:\Tuples{U}{\bbD} \rightarrow K$ and the selection predicate
$P$ maps each U-tuple to either 0 or 1 then $\sigma_{P}R :
\Tuples{U}{\bbD} \rightarrow K$ is defined by
$(\sigma_{P}R)(t) = R(t) \cdot P(t)$.
\item[Projection] If $R:\Tuples{U}{\bbD} \rightarrow K$
and $U' \subseteq U$ then
$\Pi_{U'}R:\Tuples{U'}{\bbD} \rightarrow K$ is defined
by  $(\Pi_{U'}R)(t) = \sum\sK  R(t')$ where the $+_{\sK}$ sum is over all
$t'\in\supp(R)$ such that $t'|_{U'}=t$.
\item[Natural Join] If $R_{i}:\Tuples{U_i}{\bbD} \rightarrow K,~i=1,2$
then $R_1\bowtie R2:\Tuples{U_1\cup U_2}{\bbD}\rightarrow K$ is defined by
$(R_1 \bowtie R_2)(t) = R_1(t_1) \cdot_{\sK}  R_2(t_2)$ where $t_i=t|_{U_i},~i=1,2$.
\item[Union] If $R_{i}:\Tuples{U}{\bbD} \rightarrow K,~i=1,2$
then $R_1\cup_{\sK} R_2:\Tuples{U}{\bbD}\rightarrow K$
is defined by $(R_1\cup_{\sK}  R_2)(t) = R_1(t)+_{\sK}R_2(t)$.
\end{description}

\section{Properties of $K\otimes M$}
\label{AppendixTensor}
We show here that the $K\otimes M$ constructed in Section \ref{Tensor} forms a $K$-semimodule, and highlight some of its
basic properties.

\begin{proposition}
$K\otimes M$ is a $K$-semimodule.
\end{proposition}

\begin{proof}

We show that the six semimodule axioms (definition \ref{SemiModuleDef}) hold. Four of them hold already for bags of simple tensors. For example
\eat{
$$
\begin{array}{l}
k *_{\sKM}(\ssum k_i\notimes m_i ~+_{\sKM}~\ssum k_j\notimes m_j)  =\\
\ssum (k\cdot_{\sK}k_i)\notimes m_i ~+_{\sKM}~\ssum (k\cdot_{\sK}k_j)\notimes m_j \\
\end{array}
$$
and
$$
\begin{array}{l}
k *_{\sKM} 0_{\sKM} = 0_{\sKM}
\end{array}
$$
are already instances of the definition of $*_{\sKM}$.
Moreover, the definition implies
}
$$
(k\cdot_{\sK}k') *_{\sKM}\ssum k_i\notimes m_i =
\ssum (k\cdot_{\sK}k'\cdot_{\sK}k_i)\notimes m_i =
$$
$$
= k *_{\sKM} \ssum (k'\cdot_{\sK}k_i)\notimes m_i =
k *_{\sKM} (k' *_{\sKM}\ssum k_i\notimes m_i)
$$
\eat{
and
\begin{tabbing}
$1_{\sK} *_{\sKM} \ssum k_i\notimes m_i =$ \\
$\ssum (1_{\sK}\cdot_{\sK}k_i)\notimes m_i=$ \\
$\ssum k_i\notimes m_i$
\end{tabbing}
}

By taking the quotient by congruence defined in Section~\ref{Tensor} we also get the remaining two axioms:
$$
(k+_{\sK}k') *_{\sKM} \ssum k_i\notimes m_i =
\ssum (k\cdot_{\sK} k_i +_{\sK} k'\cdot_{\sK} k_i)\notimes m_i \sim
$$
$$
\sim
\ssum (k\cdot_{\sK} k_i)\notimes m_i  +_{\sKM} (k'\cdot_{\sK} k_i)\notimes m_i =
$$
$$
= \ssum (k\cdot_{\sK} k_i)\notimes m_i
            +_{\sKM}\ssum (k'\cdot_{\sK} k_i)\notimes m_i=
$$
$$
 = k *_{\sKM}\ssum k_i\notimes m_i  +_{\sKM} k' *_{\sKM}\ssum k_i\notimes m_i
$$
and
$$
0_{\sK} *_{\sKM} \ssum k_i\notimes m_i =
\ssum 0_{\sK}\notimes m_i \sim
\ssum 0_{\sKM} = 0_{\sKM}
$$
This concludes the proof.
\end{proof}

Furthermore, $K\notimes M$ is the ``most economical" K-semimodule in the following sense. Define $\iota:M\rightarrow K\notimes M$ such that
$\iota(m)=1_{\sK} \notimes m$. Every tensor is a linear (with respect to $K$)
combination of simple tensors from $\iota(M)$. More precisely,
\begin{proposition}
\label{universalityTensor}
$K\otimes M$ satisfies the
  ``universality'' property, i.e.
for any $K$-semimodule $W$ and any homomorphism of monoids
$f:M\rightarrow W$ there exists a unique homomorphism of
$K$-semimodules $f^*:K\notimes M \rightarrow W$ such that $f^*\circ\iota = f$.
\end{proposition}

\begin{proof}
Define $f^*$ first on bags of simple tensors as follows
$$
f^*(\ssum k_i\notimes m) = \ssum k_i*_{\sW}f(m)
$$
where the second sum is taken with $+_{\sW}$ (in particular, the empty
such sum is by convention $0_{\sW}$).
Thus
$f^*(\iota(m)) = f^*(1_{\sK}\notimes m) = 1_{\sK}*_{\sW}f(m) = f(m)$.
Then, one can check that $f^*$ is a homomorphism with respect to
$+_{\sKM},0_{\sKM}$ and $*_{\sKM}$. This implies that for $f^*$ to
preserve $\sim$ it suffices
to preserve the four laws of the congruence given in Section \ref{Semimodules}, which is readily checked, for example
$$
f^*(k\notimes m +_{\sKM} k\notimes m') = k *_{\sW} f(m) +_{\sW} k *_{\sW} f(m') =
$$
$$
= k *_{\sW} (f(m)+_{\sW}f(m')) = k *_{\sW} f(m+_{\sM}m') =
f^*(k\notimes(m+_{\sM}m'))
$$
Since $f^*$ preserves $\sim$ it can be defined as above by picking a
representative from each equivalence class.
Now let $g:K\notimes M\rightarrow W$ be another linear function
such that $g\circ\iota = f$. Then
$$
g(\ssum k_i\notimes m) = g(\ssum k_i*_{\sKM}(1_{\sK}\notimes m)) =
\ssum k_i*_{\sW}g(1_{\sK}\notimes m) =
$$
$$
= \ssum k_i*_{\sW}f(m) = f^*(\ssum k_i\notimes m)
$$
hence $g=f^*$, thus verifying the uniqueness of $f^*$.
In particular, any linear function on $K\otimes M$ is completely
determined by its behavior on the tensors in $\iota(M)$.
\end{proof}

We also say that $K\otimes M$ is the $K$-semimodule {\em freely generated}
by $M$ (thus the ``most economical'' appellation).

Recall that we have defined in Section \ref{Tensor} the ``lifting'' of a homomorphism of semirings
$h: K\rightarrow K'$ to a homomorphism of monoids
$h^M: K\otimes M \rightarrow K'\otimes M$. Its definition is an  immediate consequence of Proposition~\ref{universalityTensor}:
indeed $K'\otimes M$ becomes a $K$-semimodule via $h$
so we can define $h^M$ as the the unique homomorphism
of $K$-semimodules that by the proposition
extends $\iota':M\rightarrow K'\notimes M$. Note that this yields the definition in Section \ref{Tensor}:
$$
h^M(\ssum k_i\notimes m_i) = \ssum h(k_i)\notimes m_i
$$

\section{Additional Proofs}

\begin{proof} (Proposition \ref{uniform})

The ``if'' direction
follows from the fact that homomorphisms, by definition, preserve
$\{+,\cdot,0,1\}$-expressions. For the ``only if'' direction we
use {\em abstractly tagged}~\cite{GKT-PODS07} databases.
These are $\NX$-databases in which each tuple is annotated by just
an indeterminate in $X$, and a different one at that. It is as if
tuples are annotated by their distinct id. Clearly, there is a canonical
way of choosing a large enough $X$ and producing a canonical
abstractly tagged database that is completely determined by its
support. Now, fix an operation $\Omega$, and consider its semantics on
$K$-databases, for various $K$. For any $K$ and any $K$-database $D$,
let $D^a$ be the abstractly tagged database determined by
$\supp(D)$. Let $h:\NX\rightarrow K$ be the homomorphism uniquely determined
by mapping the abstract tags in $X$ to the actual annotations of $D$.
Fix a tuple $t$ in $\supp(\Omega(D^a))$ and let
$p_t=\Omega(D^a)(t)\in\NX$ be the polynomial that
annotates it in $\Omega(D^a)$. By commutation with homomorphisms
and the definition of $\hRel{h}$ we have
$\Omega(D)(t) = \Omega(\hRel{h}(D^a))(t) = \hRel{h}(\Omega(D^a))(t)
= h(p_t)$. Using the homomorphism properties we move $h$ inside
$p_t$ until it applies to just indeterminates. For an indeterminate, say $x$,
$h(x)$ is the $K$-annotation of the unique tuple in $\supp(D)$ that is
annotated with $x$ in $D^a$. It follows that $\Omega(D)(t)$ is given by
the $\{+,\cdot,0,1\}$-expression $p_t$ in terms of the annotations of $D$.
But $p_t$ only depends on $\supp(D)$ and $\Omega$ while it is the same
for all $K$. This shows the algebraic uniformity of the semantics.
\end{proof}

\paragraph*{Commutation with homomorphism for $\SPJUlastAGB$ queries} We next prove that the semantics proposed for restricted aggregation queries (in Section 3) satisfies commutation with homomorphisms:

\begin{proof}
The proof is by induction on the query structure, but since commutation with homomorphisms was already shown for SPJU queries, we need only to prove that GB commutes with homomorphisms as well.
Let $R$ be a $\Rel{K}$ on the set of attributes $U$, where $U',U''\subseteq U$ and $U'\cap U'' = \emptyset$. $R$ may be the result of applying any sequence of $SPJU$ operations that appear in the query $Q$, followed by $GB_{U',U''}(R)$.

Consider the result when first applying the $GB$ operation followed by $\hRel{h}$. According to definition \ref{DefAgg}, the result of applying GB will be a relation $R'$, whose support contains every tuple $t$ such that:
\begin{enumerate}
\item $t$ is defined on the attributes $U'\cup U''$;
\item For some non-empty subset $T=\{t'_1,...,t'_n\}$ of $\supp(R)$, the restriction of $t$ to attributes of $U'$ is equal to the restriction of every tuple $t'_i\in T$ to $U'$, and not equal to the restriction to $U'$ of any other tuple in $\supp(R) - T$;
\item For each $u\in U''$, $t(u) = \Sigma_{t'_{i}\in T} R(t'_i)\otimes t'_i(u)$; and
\item $R'(t) = \delta\left(\Sigma_{t'_{i}\in T} R(t'_i)\right)$.
\end{enumerate}
The effect of applying $\hRel{h}$ on such $t$ would then be:
\begin{enumerate}
\item $\left(\hRel{h}(R')\right)(t)$ is equal by definition to $h(R'(t)) =$\\* $h\left(\delta\left(\Sigma_{t'_{i}\in T} R(t'_i)\right)\right) = \delta\left(\Sigma_{t'_{i}\in T} h\left(R(t'_i)\right)\right)$.
\item For every $u\in U''$, $t(u)$ will be replaced by $h^M(t(u)) = h^M\left(\Sigma_{t'_{i}\in T} R(t'_i)\otimes t'_i(u)\right) = \Sigma_{t'_{i}\in T} h(R(t'_i))\otimes t'_i(u)$.
\item The rest of the values in $t$ remain intact.
\end{enumerate}

Then, let us check the result of applying $\hRel{h}$ before the GB operation, and compare it to the above result.
Applying $\hRel{h}$ on $R$ will only affect the tuple provenance annotations; for every tuple $t'$, $\left(\hRel{h}(R)\right)(t')= h(R(t))$.
Now let us apply $GB_{U',U''}(\hRel{h}(R))$. Again, according to our semantics, the result will be a relation $R''$, whose support contains every tuple $t$ such that:
\begin{enumerate}
\item $t$ is defined on the attributes $U'\cup U''$ (as before);
\item For some non-empty subset $T=\{t'_1,...,t'_n\}$ of \\*$\supp(\hRel{h}(R))$, the restriction of $t$ to attributes of $U'$ is equal to the restriction of every tuple $t'_i\in T$ to $U'$, and not equal to the restriction to $U'$ of any other tuple in $\supp(\hRel{h}(R)) - T$;
\item for each $u\in U''$, $t(u) = \Sigma_{t'_{i}\in T} \left(\hRel{h}(R)(t'_i)\right)\otimes t'_i(u) =  \Sigma_{t'_{i}\in T} h(R(t'_i))\otimes t'_i(u)$; and
\item $R'(t) = \delta\left(\Sigma_{t'_{i}\in T} h(R(t'_i))\right)$.
\end{enumerate}

We now need to verify that these results are indeed equivalent. Note first that applying $\hRel{h}$ on $R$ before applying the GB operation, only affects the tuple annotations, and not their values. We then employ a ``by-case" analysis to verify equivalence. For any tuple $t$ that is both in the support of $R$ and in the support of $\hRel{h}(R)$, it is easy to observe from the above equations that the ``contribution" of $t$ to both the aggregation value and its annotation is the same in both cases. Consequently we can focus on tuples in $\supp(R)$ for which $\hRel{h}$ sets their annotations to $0$, thus deleting them from their groups or even deleting a whole group by deleting all its members. Those tuples and groups contribute to the result in $R'$, when applying GB first (before applying $\hRel{h}$). However, this means that the summands corresponding to the annotations of those tuples in the $\delta$ annotation of the groups in $R'$ will be later set to 0 by $\hRel{h}$; as for the aggregation results, for every tuple $t'$ that is deleted by $h$, its summand $h(R(t'))\otimes t'(u)$ in each aggregation result will be set to $0\otimes t'_i(u)=0_{\sKM}$ and thus it will have no affect on the aggregation results. We thus conclude that if no group has been deleted altogether, the results are equivalent.  The last case to consider is that where all the annotations of tuples in group $T$ are set to zero. In this case, its $\delta$ expression will be equal to zero as well, so the group is effectively deleted also by $\hRel{h}$ after the GB.

This concludes the proof.
\end{proof}

\begin{proof} (Thm. \ref{compatibilityIdempotence})

Let $K$ be a commutative semiring which is positive with respect to $+_{\sK}$ and define $h: K \notimes M \rightarrow M$ as $h(\sum_{i \in I} k_{i} \notimes m_i) = \sum_{j \in J} m_i$ where $J=\{j \in I \mid k_{j} \neq 0\}$. We can show that $h$
is well-defined (see below); since $\forall m \in M~~ h \circ \iota(m) = m$, $\iota$ is injective and thus $K$ and $M$ are compatible.

We need to verify that $h$ is a well-defined mapping, and for that we check that it is well-defined on
$K \notimes M$ after taking the quotient (as defined in Section \ref{Semimodules}):

\begin{itemize}
\item (For $k,k'\neq 0_{\sK}$) $h((k+_{\sK} k') \otimes m) = m= m+_{\sM}m = h(k \otimes m +_{\sKM} k' \otimes m)$.
\item $h(0_{\sK} \otimes m)=0_{\sM}$, and also the empty bag is mapped to the ``empty sum" i.e. $h(0_{\sKM}) = 0_{\sM}$.
\item (For $k\neq 0_{\sK}$) $h(k \otimes (m+_{\sM}m'))= m+_{\sM}m' = h(k \otimes m +_{\sKM} k \otimes m')$.
\item $h(k \otimes 0_{\sM}) = 0_{\sM}$, and again $h(0_{\sKM}) = 0_{\sM}$.
\end{itemize}

Note that we assumed that $k$ and $k'$ are non-zero in the first and third axioms. Since $K$ is positive with respect to $+_{\sK}$, no such $k,k'$ can satisfy $k+_{\sK}k'=0_{\sK}$, thus the case of $0_{\sK} \otimes m$ is \emph{uniquely} defined to be mapped to $0_{\sM}$, by the second axiom.

This concludes the proof.
\end{proof}

\begin{proof} (Thm. \ref{compatibilityHomomorphism})

Let $h'$ be a homomorphism from $K$ to $\bbN$, and $M$ be an arbitrary commutative monoid. We define a mapping $h:K\otimes M \rightarrow M$, as follows.\\
$h\left(\Sigma k_i\otimes m_i\right) = \Sigma h'(k_i)m_i$. We show that $h$ is well-defined and that $h\circ \iota$ is the identity function.

We first show that this mapping is well-defined, i.e. that every pair of elements from $K\otimes M$ which are equated by the axioms of the tensor construction
(as defined in Section \ref{Semimodules}), are mapped by $h$ to the same values.
\begin{description}
\item[1st axiom.] Left side: $h\left((k+_{\sK} k')\otimes m\right)$ is equal by the definition of $h$ to $h'(k+_{\sK} k')m$. Since $h'$ is a homomorphism, this is equal to $\left(h'(k)+h'(k')\right)m$. Right side: $h\left(k\otimes m+_{\sKM} k'\otimes m\right)= h'(k)m +_{\sM} h'(k')m$ by $h$ definition. Since $h'(k), h'(k')$ are natural numbers, this is equal to the result of the left hand side.
\item[2nd axiom.] Left side: $h\left(0_{\sK}\otimes m\right) = h'(0_{\sK})m$. Since $h'$ is a homomorphism, $h'(0_{\sK})=0$ and thus the expression is equal to $0m = 0_{\sM}$. Right side: by definition of $h$, the ``empty'' sum in $K\otimes M$ must be mapped to the ``empty'' sum in $M$, which is $0_{\sM}$.
\item[3rd axiom.] Left side: $h\left(k\otimes (m+_{\sM} m')\right)=h'(k)(m+_{\sM} m') = h'(k)m +_{\sM} h'(k)m'$. Right side: \\* $h\left(k\otimes m+_{\sKM} k\otimes m'\right)=$ $h'(k)m +_{\sM} h'(k)m'$.
\item[4th axiom.] Left side: $h\left(k\otimes 0_{\sM}\right)=h'(k)0_{\sM} = 0_{\sM}$. Right side: same as the 2nd axiom.
\end{description}

Since $h$ is well-defined such that $h(a+b)=h(a)+h(b)$, it is a homomorphism from $K\otimes M$ to $M$.
Now we need to show that $h \circ \iota$ is the identity function, implying that $\iota$ is injective and thus that $M$ and $K$ are compatible. This is easy: since $h'$ is a homomorphism it must map $1_{\sK}$ to $1$; then for every $m\in M$, $h(\iota(m)) =h(1_{\sK} \otimes m) = h'(1_{\sK})m =1m = m$.
\end{proof}

\begin{proof} (Proposition \ref{noMKGenAgg})

%Let $D$ be the database
We show the proof for $\SumAgg$~and the proof for $\MaxAgg$ ($\MinAgg$) is similar. Reconsider the relation $R'$ and the query $Q_{select}$ in Example \ref{ImpossibleExample}, and assume that $R''=Q_{select}(R')$ is a $(M,\NX)$-relation capturing the query result according to some algebra. Assume by contradiction that the algebra commutes with homomorphism. Let $h, h'$ be homomorphisms from $\NX$ to $\bbN$ corresponding to those in Example \ref{noMK}, I.e. $h(r_1)=h(r_3)=h'(r_1)=h'(r_2)=h'(r_3)=1$ and $h(r_2)=0$. We saw in Example \ref{ImpossibleExample} that the aggregation result when $h$ is applied is 20; thus, in order to be bag-compatible, $h^M(R'')$ must include a tuple $t''$ representing this tuple which matched the selection condition. Let $p_{t''}\in \NX$ be its provenance annotation, then $h(p_{t''})=1$. However, $h'^M(R'')$ is empty, since no aggregation result is equal to 20 in that case. Thus $h'(p_{t''})=0$. Similarly to the proof of prop. \ref{noMK}, observe that there exists no such polynomial $p_{t''}\in \NX$.
\end{proof}
%In the proof of Theorem \ref{noMK}, we showed it is impossible to encode the two possible aggregation results, 20 and 30 in different tuples. Thus, $R''$ must contain a single tuple whose value

%TBD: Adapt the proof.

%It is easy to observe that there are two options for $D''$, as follows. First, it may contain two tuples as follows (with some provenance annotations):
%
%$$
%\begin{array}{|cccc|}
%%\multicolumn{4}{c}{R} \\
%%\multicolumn{4}{c}{} \\
%\hline
%\mathit{A} & \mathit{SumB} & \mathit{D} & \mathit{E} \\
%\hline
%\hline
%a & 1 & 1 & 200 \\
%a & 2 & 2 & 100 \\
%
%\hline
%\end{array}
%$$
%
%The annotations of $D''$ must ensure that a homomorphism that sets $t_1=t_2=t_3=t_4=1$ would lead to a non-zero annotation for the first tuple, and a zero annotation for the second. Also,  a homomorphism that sets $t_1=t_2=t_3=1$ but $t_4=0$ would lead to a zero annotation for the first tuple, and non-zero annotation for the second. This is impossible due to the monotonicity of $\NX$.
%
%The second option is that $D''$ contains a single tuple, but uses some elements of $K \otimes M$ for the values. Specifically consider $k \otimes m$ that is in the $E$ attribute of the single tuple of $D''$. Then it must be the case that for a homomorphism $h$ ($h'$) that sets $t_1=t_2=t_3=t_4=1$, ($t_1=t_2=t_3=1$ and $t_4=0$), $h(k \otimes m) = 1$ ($h'(k \otimes m) = 2$). But it must also the case that $h(k \otimes m) > h'(k \otimes m)$ ($+_{\sM}$ is just a summation of natural numbers here, and is thus monotone), thus we obtain $100>200$, showing that our initial assumption was false.
%\end{proof}

\begin{proof} (Theorem \ref{SolveMK})

The proof is by induction on the structure of elements in $K^{M}$. We say that an expression $exp \in K^{M}$ has a nesting level of $0$ if for every expression $[c_1\otimes m_1=c_2\otimes m_2]$ appearing in $exp$, $c_1,c_2 \in K$ (and $m_1,m_2 \in M$); $exp$ has a nesting level $n$ if each such $c_1,c_2$ are of nesting level $n-1$ or less. For nesting level of $0$, axiom (*) above allows us to replace $[c_1=c_2]$ with  $1_{K}$ or $0_{\sK}$.  Now, assume that the theorem holds for expressions with nesting level $n-1$ or less, and let $exp$ be of nesting level $n$. Then for each sub-expression $[c_1\otimes m_1=c_2\otimes m_2]$, we can replace, by the induction hypothesis (and using the axioms above), $c_1,c_2$ with elements of $K$ and then apply axiom (*) to replace the equality expression with an element of $K$. We can repeat for every equality sub-expression of $exp$, obtaining an element of $K$.
\end{proof}

\paragraph*{Commutation with homomorphism for the extended semantics} We next prove that the semantics proposed for nested aggregation queries (in Section 4) satisfies commutation with homomorphisms:

\begin{proof}
The proof is by induction on the query structure. For each operation we consider two cases: (I) applying the homomorphism $\hRel{h}$ \emph{after} applying the operation; (II) applying it \emph{before} the operation. Both cases must yield equal results.

In what follows we use the same notations, $R,R_1,R_2$ and so on, as used in the definition of the extended semantics.
\paragraph*{Union}
First, consider case (I), where the union is applied first. According to the defined semantics, the result of $R_1\cup R_2$ is a $(M,K^M)$-relation such that for every tuple $t$ in its support it holds that:
\begin{enumerate}
\item $t$ is defined (only) on the attributes in $U$.
\item $t\in \supp(R_1)\cup\supp(R_2)$
\item $(R_1\cup R_2)(t) = \sum_{t' \in supp(R_{1})}R_{1}(t')\cdot\!_{\sK} \prod_{u \in U}[t'(u)=t(u)]+_{\sK} \sum_{t' \in supp(R_{2})}R_{2}(t') \cdot\!_{\sK}  \prod_{u \in U}[t'(u)=t(u)]$.
\end{enumerate}
Then, applying $\hRel{h}$ on $R_1\cup R_2$ has the following effect on that $t$.
\begin{enumerate}
\item For every value in $t$ from $K^M\otimes M$, its value changes from $\Sigma k_i\otimes m_i$ to $h^M\left(\Sigma k_i\otimes m_i\right)=\Sigma h(k_i)\otimes m_i$.
\item The provenance annotation of $t$ is changed according to the axioms of the homomorphism lifting, to \\* $\hRel{h}\left((R_1\cup R_2)\right)(t) = $\\*$\sum_{t' \in supp(R_{1})}h^M(R_{1}(t'))\cdot \prod_{u \in U}[h^M(t'(u))=h^M(t(u))]+ \sum_{t' \in supp(R_{2})}h^M(R_{2}(t')) \cdot \prod_{u \in U}[h^M(t'(u))=h^M(t(u))]$
\end{enumerate}

Now, for case (II), let us apply $\hRel{h}$ first, on both $R_1$ and $R_2$. This would affect only the provenance annotations of tuples within this relations, causing perhaps to the deletion of some tuples, and the values from $K^M\otimes M$, which change in the same manner as described in item (1) above. Let us compute the result of $\hRel{h}(R_1)\cup \hRel{h}(R_2)$. Every tuple $t$ in the support of the obtained relation is such that:
\begin{enumerate}
\item $t$ is defined (only) on the attributes in $U$.
\item $t\in \supp(R_1)\cup\supp(R_2)$ and it holds that either \\* $h(R_1(t))\neq 0_{\!_{K'}}$ or $h(R_2(t))\neq 0_{\!_{K'}}$
\item $\left(\hRel{h}(R_1)\cup \hRel{h}(R_2)\right)(t) = \sum_{t' \in supp(R_{1})}h^M(R_{1}(t'))\cdot \prod_{u \in U}[h^M(t'(u))=h^M(t(u))]+_{\sK} \sum_{t' \in supp(R_{2})}h^M(R_{2}(t')) \cdot \prod_{u \in U}[h^M(t'(u))=h^M(t(u))]$.
\end{enumerate}

We next verify that the results are indeed equal. for every tuple $t$ there are several options: it can be in the support of $R_1$, $R_2$, neither or both; and $h$ can set its provenance to 0 in none of them, one of them or both. Any tuple $t$ which is not in $\supp(R_1)$ or $\supp(R_2)$ clearly does not affect the result. Any tuple who is at least in one of them, will be annotated in case (I) with a sum of each annotation of each tuple $t'$ in $\supp(R_1)\cup\supp(R_2)$, multiplied by tokens that equate each value from $k^M\otimes M$ in $t$ to the value of the same attribute in $t'$. In the worst case, where all the values are from $K^M\otimes M$, we do not know which tuples are equal and thus we compare each pair on each attribute. Then applying $h$ might cause some of the original tuple annotation, hence some of the summands in the provenance of $t$ to become $0$. In case (II) those tuples for which $h$ sets the annotations to 0 are removed from $R_1$, $R_2$ or both, and thus they do not appear in the sum to begin with. For the tuples that remain it is clear that the obtained annotations and $K^M\otimes M$ are the same in both case (I) and case (II).

One thing to note here is that different tuples in a relation, $R_1$ for instance, may be equated after applying $\hRel{h}$. This is true, for instance, when we have two tuples which differ only by some aggregation result, but after applying $h$ those results turn out to be the same. This works well with the homomorphism commutation as well, because it is easy to see that in both cases the tuples will be equal, in case (I) after applying $h$ on the union result and in case (II) before the union is applied.

The proof for projection is very similar to the one for union, thus it is not repeated here.

\paragraph*{Selection}
According to the algebra definition, to get $\left(\sigma_{P}(R)\right)(t)$ we simply multiply the annotation of each tuple $t\in\supp(R)$ by an expression equating the value of the relevant attribute $u$ in $t$ to some value $m$ (embedded into $K^M\otimes M$ using $\iota$).

In case (I) the provenance of some tuple $t$ in the support of $\sigma_{P}(R)$ might be $R(t)\!\cdot_{\sK}\!\left[\left(\Sigma k_i\otimes m_i\right) = 1_{\sK}\!\otimes \!m\right]$, which would become, after applying $h$, \\* $h(R(t))\!\cdot_{\sK}\!\left[\left(\Sigma h(k_i)\otimes m_i\right) = 1_{\!_{K'}}\!\otimes\! m\right]$.

In case (II) the annotation of tuple $t$ will become $h(R(t))$ after applying $\hRel{h}$, and $t(u)$ would become $\Sigma h(k_i)\otimes m_i$. Thus after applying the selection, we would get the same result, $h(R(t))\!\cdot_{\sK}\!\left[\left(\Sigma h(k_i)\otimes m_i\right) = 1_{\!_{K'}}\!\otimes\! m\right]$.

\paragraph*{Aggregation}
In case (I), we first apply the GB operation $GB_{U',u}R$ and obtain a relation such that for every tuple $t$ in its support it holds that:
\begin{enumerate}
\item $t$ is defined (only) on the attributes in $U\cup\{u\}$.
\item There exists some tuple $t'\in\supp(R)$ such that for every attribute $u'\in U'$, $t(u')=t'(u')$.
\item $t(u) = \Sigma_{t'\in\supp(R)} \left(R(t')\!\cdot_{\sK}\!\prod_{u'\in U'}\left[t(u') = t'(u')\right]\right) \!*_{\sKM} t'(u)$
\item $GB_{U',u}R(t) = \delta(\Sigma_{t'\in\supp(R)} R(t')$\\* $\cdot_{\sK}\!\prod_{u'\in U'}\left[t(u') = t'(u')\right])$.
\end{enumerate}
Now, after applying $\hRel{h}$ on the result, the effect on such tuple $t$ would be:
\begin{enumerate}
\item $t(u) = \Sigma_{t'\in\supp(R)} (h^M(R(t'))$ \\* $\cdot_{\sK}\!\prod_{u'\in U'}\left[h^M(t(u')) = h^M((t'(u'))\right]) \!*_{\sKM} h^M(t'(u))$
\item $\hRel{h}(GB_{U',u}R)(t) = \delta(\Sigma_{t'\in\supp(R)} h^M(R(t'))$ \\* $\cdot_{\sK}\!\prod_{u'\in U'}\left[h^M(t(u')) = h^M(t'(u'))\right])$.
\end{enumerate}

In case (II), we first apply $\hRel{h}$, which affects the tuple provenances and the values from $K^M\otimes M$. Then aggregation is applied on the result. Each tuple $t$ in \\*$\supp(GB_{U',u}\hRel{h}(R))$ is such that:
\begin{enumerate}
\item $t$ is defined (only) on the attributes in $U\cup\{u\}$.
\item There exists some tuple $t'\in\supp(\hRel{h}(R))$ such that for every attribute $u'\in U'$, $t(u')=t'(u')$.
\item $t(u) = \Sigma_{t'\in\supp(R)} (h^M(R(t'))$ \\* $\cdot_{\sK}\!\prod_{u'\in U'}\left[t(u') = h^M(t'(u'))\right]) *_{\sKM} h^M(t'(u))$
\item $GB_{U',u}\hRel{h}(R)(t) = \delta(\Sigma_{t'\in\supp(R)} h^M(R(t'))$\\*$\cdot_{\sK}\!\prod_{u'\in U'}\left[t(u') = h^M(t'(u'))\right])$.
\end{enumerate}

We now verify that the results in both cases are indeed equal. In the first case, according to the definition, every tuple $t$ in $\supp(R)$ is forming the basis of a group, which conditionally may contain every tuple in $R$ (using equation expressions to verify that each tuple is indeed in that group only if its restriction to $U'$ is equal to the restriction of $t$ to $U'$).
When we apply $h$, some tuple annotations may be set to $0$, and thus their corresponding summands (in the group provenances and aggregation results) are set to $0$, and do not affect the result. In case (II) some tuples may be removed by $h$ from the relation even before the aggregation is performed. This has a similar effect to setting their corresponding summands to 0 as in case (I). There is a slight difference here: If $\supp(R)$ was of size $n$, so will be the size of the support of $GB_{U',u}R$, maybe even after applying $\hRel{h}$ on it; However, $\supp(\hRel{h}(R))$ may be of size $m < n$, and thus so will be $GB_{U',u}\hRel{h}(R)$. The reason is that as long as we cannot evaluate equation values, we have to allow for $n$ different groups (as many as there are in the support), thus if there are ``actually" less, when we move to a $K'^M\otimes M$ where equations may be evaluated, we may get the same tuple representing the group duplicated the number of times as the number of its group members. This is acceptable, since duplicates are ignored. However if we apply $h$ first, we may know that there cannot be $n$ groups to begin with, since some of the tuples are deleted. The effect of this will only be less group duplicates. The commutation of the AGG operation with homomorphism follows from the above proof as well.

This concludes the proof.
\end{proof}

\begin{proof} (Proposition \ref{Neg})

First, note that by the homomorphism commutation, we can check the equivalence of using both expressions on $h^{\widehat{B}}(R)$ and $h^{\widehat{B}}(S)$, i.e. verify that \\* $((\Pi_{a_1...a_n}\{\left(GB_{\{a_1,...a_n\},b} (h^{\widehat{B}}(R) \times \bot_b \cup h^{\widehat{B}}(S) \times \top_b)\right)$\\*$\bowtie (h^{\widehat{B}}(R) \times \bot_b)\})(t)) = [h^{\widehat{B}}(S)(t) \otimes \top = 0] \!\cdot\!_{\!_{K'}}\!h^{\widehat{B}}(R)(t)$. Note that since $K' \otimes \widehat{B}$ is isomorphic via an isomorphism $I$  to $\widehat{B}$, instead of checking equality in $K' \otimes \widehat{B}$ we can check for equality after application of the isomorphism (in $\widehat{B}$); in particular this allows us to interpret the equalities and replace them with $1_{\sK}$ or $0_{\sK}$.

Now let us follow the operation of difference encoded by aggregation step-by-step. Let $R,S \in \bbD^U\rightarrow K'^{\widehat{B}}$; let \\* $\supp(h^{\widehat{B}}(R)) = \{r_1,...,r_n\}$ and $\supp(h^{\widehat{B}}(S))=\{s_1,...,s_m\}$. Then $\supp(h^{\widehat{B}}(R)\times \bot_b) = \{r'_i \mid r'_i(b)=\bot\wedge\exists r_i\in \supp(h^{\widehat{B}}(R))$\\*$\forall u\in U~~ r'_i(u)=r_i(u)\}$; since $\times$ is equivalent to join with no attribute equations, the provenance $\left(h^{\widehat{B}}(R)\times \bot_b\right)(t)$ is $h^{\widehat{B}}(R)(t)\!\cdot_{\!_{K'}}\bot_b(t)=h^{\widehat{B}}(R)(t)$ for any $t\in \supp(h^{\widehat{B}}(R)\times \bot_b)$, and $0_{\!_{K'}}$ otherwise; similarly for $h^{\widehat{B}}(S)\times \top_b$. Now, note that the support sets of both relations are mutually exclusive, and that all the values in the tuples in the support of both $h^{\widehat{B}}(R)$ and $h^{\widehat{B}}(S)\times \top_b$ are from $\bbD$. Thus, it is easy to see that $\left(h^{\widehat{B}}(R)\times \bot_b\right) \cup \left(h^{\widehat{B}}(S)\times \top_b \right)(t)$ is $h^{\widehat{B}}(R)(t|_{U})$ for $t\in \supp(h^{\widehat{B}}(R)\times \bot_b)$, $h^{\widehat{B}}(S)(t|_{U})$ for $t\in\supp(h^{\widehat{B}}(S)\times \top_b)$, and $0_{\!_{K'}}$ otherwise. Now we apply group-by on the $b$ attribute. There are four possible classes of tuples:
\begin{enumerate}
\item For every $r'_i$ such that $r_i=s_j\in h^{\widehat{B}}(R)\cap h^{\widehat{B}}(S)$, we would get a tuple $r''_i$ in the support of the result, such that $r''_i|_{U}=r'_i|_{U}=r_i=s_j=s'_j|_{U}$, and $r''_i(b) = \left(h^{\widehat{B}}(R)\times \bot_b\right) (r'_i)\otimes \bot +_{K'\otimes\widehat{B}} h^{\widehat{B}}(S)\times \top_b (s'_i)\otimes \top = h^{\widehat{B}}(S)(r_i)\otimes \top$. The provenance of this tuple would be $\delta(h^{\widehat{B}}(R)(r_i)+_{\!_{K'^{\widehat{B}}}}h^{\widehat{B}}(S)(r_i))$.
\item For every $r'_i$ such that $r_i\in h^{\widehat{B}}(R) \backslash h^{\widehat{B}}(S)$, there is a tuple $r''_i$ in the group-by result such that $r''_i(b)= \left(h^{\widehat{B}}(R)\times \bot_b (r'_i)\right)\otimes \bot = 0_{\!_{K'\otimes\widehat{B}}}$, and where the provenance of $r''_i$ is $\delta(h^{\widehat{B}}(R)(r_i))$.
\item For every $s'_i$ such that $s_i\in h^{\widehat{B}}(S)\backslash h^{\widehat{B}}(R)$, there is a tuple $s''_i$ in the group-by result such that $s''_i(b)= \left(h^{\widehat{B}}(S)\times \top_b (s'_i)\right)\otimes \top = h^{\widehat{B}}(S)(s_i)\otimes \top$, and where the provenance of $s''_i$ is $\delta(h^{\widehat{B}}(S)(s_i))$.
\item Every other tuple has provenance $0_{\!_{K'}}$, i.e. it is not in the support of the aggregation result.
\end{enumerate}
Now we need to perform a join of the group-by result and $h^{\widehat{B}}(R)\times \bot_b$. For that we will first rename the attributes of the aggregation result to $U_2\cup \{b'\}$. For each tuple $t$ in the aggregation result such that $t|_{U}\in h^{\widehat{B}}(R)$, i.e. tuples of cases (1) and (2) above, there will be a unique corresponding tuple in the join result (since all these tuples have unique $U$ values, they will join with a unique tuple in $h^{\widehat{B}}(R)\times \bot_b$). The obtained provenance of the join of tuples from case (1), would be $\delta(h^{\widehat{B}}(R)(r_i)+h^{\widehat{B}}(S)(r_i))\!\cdot\!h^{\widehat{B}}(R)(r_i)\!\cdot\![h^{\widehat{B}}(S)(r_i)\otimes \top= 0]$. However, $\delta(h^{\widehat{B}}(R)(r_i)+h^{\widehat{B}}(S)(r_i))$ is redundant here: if $h^{\widehat{B}}(S)(r_i)\neq 0_{\!_{K'}}$ or $h^{\widehat{B}}(R)(r_i)=0_{\!_{K'}}$, the provenance is $0_{\!_{K'}}$; otherwise, $\delta(h^{\widehat{B}}(R)(r_i)+h^{\widehat{B}}(S)(r_i))=\delta(h^{\widehat{B}}(R)(r_i))=1_{\!_{K'}}$. As for the join of tuples from case (2), the case is simpler: the provenance is $\delta(h^{\widehat{B}}(R)(r_i))\!\cdot\!h^{\widehat{B}}(R)(r_i)\!\cdot\![0 = 0] = \delta(h^{\widehat{B}}(R)(r_i))\!\cdot\!h^{\widehat{B}}(R)(r_i)$, and again $\delta(h^{\widehat{B}}(R)(r_i))$ has no effect. For cases (3) and (4) the provenance is $0_{\!_{K'}}$. This matches the result of $[(h^{\widehat{B}}(S))(t) \otimes \top = 0] \!\cdot\!_{\!_{K'}}\!(h^{\widehat{B}}(R))(t)$.
\end{proof}

\end{document}